%% file: main-arxiv.tex
\documentclass[11pt]{article}
\usepackage[letterpaper,margin=1in]{geometry}
\usepackage{amsfonts,amsmath,amssymb,amsthm,hyperref,subcaption,lineno,graphicx}
\usepackage[capitalise]{cleveref}
\usepackage{thm-restate}
\usepackage[labelfont=bf,font={small}]{caption}
\usepackage{authblk}

\usepackage{microtype,verbatim,dsfont}
\usepackage{xspace,gensymb}
\usepackage{xcolor}
\usepackage{complexity}
\usepackage{paralist}
\usepackage{wrapfig}

\Crefname{section}{Section}{Sections}
\Crefname{lemma}{Lemma}{Lemmas}
\Crefname{property}{Property}{Properties}
\Crefname{enumi}{Condition}{Conditions}
\Crefname{figure}{Figure}{Figures}

\theoremstyle{plain}
\newtheorem{theorem}{Theorem}
\newtheorem{lemma}{Lemma}
\newtheorem{definition}{Definition}
\newtheorem{claim}{Claim}
\newtheorem{remark}{Remark}
\newtheorem{property}{Property}

\newcommand{\mfunc}[1]{\mathcal M(#1)}
\newcommand{\compact}{compact\xspace}
\newcommand{\sketch}{sketch\xspace}
\newcommand{\sketches}{sketches\xspace}
\newcommand{\feasible}{feasible\xspace}
\newcommand{\de}[1]{\textrm{deg(\ensuremath{#1})}}
\newcommand{\construction}{\compact construction\xspace}

\title{Recognizing Map Graphs of Bounded Treewidth\thanks{An extended abstract of this manuscript has appeared in the proceedings of the 18th Scandinavian Symposium and Workshops on Algorithm Theory (SWAT 2022) \cite{DBLP:conf/swat/AngeliniBLGMT22}.}}

\author[1]{Patrizio~Angelini}

\author[2]{Michael~A.~Bekos}

\author[3]{Giordano~Da~Lozzo}

\author[4]{Martin~Gronemann}

\author[5]{Fabrizio~Montecchiani}

\author[5]{Alessandra~Tappini}

\affil[1]{Department of Mathematics, Natural, and Applied Sciences, John Cabot University, Rome, Italy. \texttt{pangelini@johncabot.edu}}

\affil[2]{Department of Mathematics, University of Ioannina, Ioannina, Greece. \texttt{bekos@uoi.gr}}

\affil[3]{Department of Engineering, Roma Tre University, Rome, Italy. \texttt{giordano.dalozzo@uniroma3.it}}

\affil[4]{Algorithms and Complexity Group, TU Wien, Vienna, Austria. \texttt{mgronemann@ac.tuwien.ac.at}}

\affil[5]{Department of Engineering, University of Perugia, Perugia, Italy. \texttt{fabrizio.montecchiani@unipg.it, alessandra.tappini@unipg.it}}

\date{}

\begin{document}

\maketitle

\begin{abstract}
A map graph is a graph admitting a representation in which vertices are nations on a spherical map and edges are shared curve segments or points between nations. We present an explicit fixed-parameter tractable algorithm for recognizing  map graphs parameterized by treewidth. The algorithm has time complexity that is linear in the size of the graph and, if the input is a yes-instance, it reports a certificate in the form of a so-called witness. Furthermore, this result is developed within a more general algorithmic framework that allows to test, for any $k$, if the input graph admits a $k$-map (where at most $k$ nations meet at a common point) or a hole-free~$k$-map (where each point of the sphere is covered by at least one nation).
We point out that, although bounding the treewidth of the input graph also bounds the size of its largest clique, the latter alone does not seem to be a strong enough structural limitation to obtain an efficient time complexity. In fact, while the largest clique in a $k$-map graph is $\lfloor 3k/2 \rfloor$, the recognition of $k$-map graphs is still open for any fixed $k \ge 5$.
\end{abstract}

\section{Introduction}

Planarity is one of the most influential concepts in Graph Theory. Inspired by topological inference problems and by intersection graphs of planar curves, in 1998, Chen, Grigni and Papadimitriou~\cite{DBLP:conf/stoc/ChenGP98} suggested the study of map graphs as a generalized notion of planarity. A \emph{map} of a graph $G$ is a function $\mathcal M$ that assigns each vertex $v$ of $G$ to a region $\mfunc{v}$ on the sphere homeomorphic to a closed disk such that no two regions share an interior point, and any two distinct vertices $v$ and $w$ are adjacent in $G$ if and only if the boundaries of $\mfunc{v}$ and $\mfunc{w}$ share at least one point. For each vertex $v$ of $G$, the region $\mfunc{v}$ is called the \emph{nation} of $v$.  A connected open region of the sphere that is not covered by nations is a \emph{hole}. 
A graph that admits a map is a \emph{map graph}, whereas a graph that admits a map without holes is a \emph{hole-free map graph}; \cref{fig:map-0,fig:map-1} show a graph and a map of it, respectively. Map graphs generalize planar graphs by allowing local non-planarity at points where more than three nations meet. In fact, the planar graphs are exactly those graphs having a map in which at most three nations share a boundary point~\cite{DBLP:conf/stoc/ChenGP98,DBLP:journals/siamdm/DujmovicEW17}. 

Besides their theoretical interest, the study of map graphs is motivated by applications in graph drawing, circuit board design, and topological inference problems~\cite{DBLP:journals/jgaa/AngeliniLBFPR17,DBLP:journals/dam/Brandenburg19,DBLP:journals/algorithmica/Brandenburg19,DBLP:conf/soda/ChenHK99}. 
Map graphs are also useful to design parameterized and approximation algorithms for several optimization problems that are \NP-hard on general graphs~\cite{DBLP:journals/jal/Chen01a,DBLP:journals/talg/DemaineFHT05,DBLP:conf/focs/FominLMS12,DBLP:conf/icalp/FominLP0Z19,DBLP:conf/soda/FominLS12}.

A natural and central algorithmic question regards the existence of efficient algorithms for recognizing map graphs. Towards an answer to this question, Chen et al.~\cite{DBLP:conf/stoc/ChenGP98,DBLP:journals/jacm/ChenGP02} first gave a purely combinatorial characterization of map graphs: A graph is a map graph if and only if it admits a witness,  formally defined as follows; see \cref{fig:map-2}. A \emph{witness} of a graph $G = (V,E)$ is a bipartite planar graph $W = (V \cup I,A)$ with $A \subseteq V \times I$ and such that $W^2[V] = G$, where the graph $W^2[V]$ is the \emph{half-square} of $W$, that is, the graph on the vertex set~$V$ in which two vertices are adjacent if and only if their distance in $W$ is 2. 
Here, the vertices in $I$ are meant to represent the adjacencies among nations.
Since $W$ can always be chosen to have linear size in the number of vertices of $G$~\cite{DBLP:journals/jacm/ChenGP02}, the problem of recognizing map graphs is in \NP. In 1998,   Thorup~\cite{DBLP:conf/focs/Thorup98} proposed a  polynomial-time algorithm to recognize map graphs. However, the extended abstract by Thorup does not contain a complete proof of the result and, to the best of our knowledge, a full version has not appeared yet. Moreover, the proposed algorithm has two  drawbacks. First, the time complexity is not specified explicitly (the exponent of the polynomial bounding the time complexity is estimated to be about 120~\cite{DBLP:journals/algorithmica/ChenGP06}; see also~\cite{DBLP:journals/algorithmica/Brandenburg19,DBLP:journals/disopt/MnichRS18}). Second, it does not report a certificate in the positive case; a natural one would be a witness.

Hence, the problem of finding a simple and efficient recognition algorithm for map graphs remains open. In the last years, several authors focused on graphs admitting restricted types of maps. Aside from the already defined hole-free maps, another notable example consists of the \emph{$k$-maps}, in which at most $k$ nations meet at a common point; observe that, when $k\geq n-1$, map graphs and $k$-map graphs trivially coincide.
For instance, Chen studied the density of $k$-map graphs \cite{DBLP:journals/jgt/Chen07}. As another example,  
in a recent milestone paper on linear layouts, Dujmovi\'c et al.~\cite{DBLP:journals/jacm/DujmovicJMMUW20} proved that the queue number of $k$-map graphs is cubic in $k$; this bound has been recently improved to linear~\cite{DBLP:journals/corr/abs-2204-11495}. 
Note that the algorithm by Thorup~\cite{DBLP:conf/focs/Thorup98} cannot be directly used to recognize $k$-map graphs  (unless $k \geq n-1$). Chen et al.~\cite{DBLP:journals/algorithmica/ChenGP06} focused on hole-free $4$-map graphs and gave a cubic-time recognition algorithm for this graph family. 
Later, Brandenburg~\cite{DBLP:journals/algorithmica/Brandenburg19} gave a cubic-time recognition algorithm for general (i.e., not necessarily hole-free) $4$-map graphs, by exploiting an alternative characterization of these graphs closely related to maximal $1$-planarity. 
Notably, a polynomial-time recognition algorithm for the family of (general or hole-free) $k$-map graphs with $k>4$ is still missing.  
In particular, for $k>4$, the only result we are aware of is a characterization of $5$-map graphs in terms of forbidden crossing patterns~\cite{DBLP:journals/dam/Brandenburg19}. A different approach for the original problem is the one by Mnich, Rutter, and Schmidt~\cite{DBLP:journals/disopt/MnichRS18}, who proposed a linear-time algorithm to recognize the map graphs with an outerplanar witness, which also reports a certificate witness, if any. 

We remark that the size of the largest clique in a $k$-map graph is $\lfloor 3k/2 \rfloor$ (see, e.g.,~\cite{DBLP:journals/jacm/ChenGP02}), thus bounding the size of the largest clique does not seem to be a strong enough structural limitation of the input to obtain an efficient time complexity. Despite the notable amount of work, no prior research focuses on further structural parameters of the input graph to design efficient recognition algorithms. In this paper, we address precisely this challenge.

\begin{figure}
\centering
\begin{subfigure}[b]{.3\textwidth}
\centering
\includegraphics[page=10,width=\textwidth]{figs/compact-map}
\subcaption{A graph $G$.\label{fig:map-0}}
\end{subfigure}\hfill
\begin{subfigure}[b]{.3\textwidth}
\centering
\includegraphics[page=1,width=\textwidth]{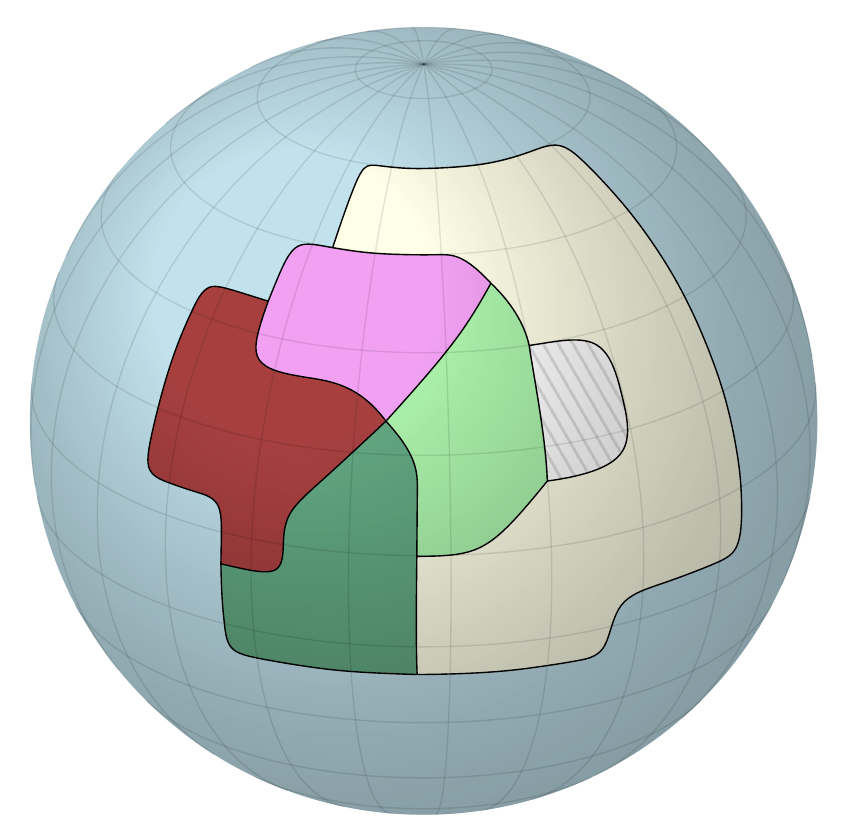}
\subcaption{A map $\cal M$ of $G$.\label{fig:map-1}}
\end{subfigure}\hfill
\begin{subfigure}[b]{.3\textwidth}
\centering
\includegraphics[page=2,width=\textwidth]{figs/sphere_tikz}
\subcaption{A witness $W$ of $G$.\label{fig:map-2}}
\end{subfigure}
\caption{(a) A graph $G$, (b) a map of $G$ - the striped region is a hole, and (c) a witness of~$G$.\label{fig:map}}
\end{figure}
\medskip\noindent\textbf{Our contribution.} Our main result is a novel algorithmic framework that can be used to recognize map graphs, as well as variants thereof; in particular,  hole-free $k$-map graphs and $k$-map graphs.  Recall that, by setting $k=n-1$, our algorithm also recognizes (hole-free) map graphs. In fact, we can also compute the minimum value of $k$ within the same asymptotic running time. The proposed algorithm is parameterized by the treewidth~\cite{DBLP:series/mcs/DowneyF99,DBLP:journals/jal/RobertsonS86} of the $n$-vertex input graph $G$ and its time complexity has a linear dependency in~$n$, while it does not depend on the natural parameter~$k$. Notably, for graphs of bounded treewidth, our algorithm improves over the existing  literature~\cite{DBLP:journals/algorithmica/Brandenburg19,DBLP:journals/algorithmica/ChenGP06,DBLP:conf/focs/Thorup98} in three ways: it solves the problem for {\em any} fixed~$k$, it can deal with {\em both} scenarios where holes are or are not allowed in the sought map, and it exhibits an asymptotically {\em optimal} running time in the input size. The following theorem summarizes our main contribution. 

\begin{theorem}\label{th:main}
Given an $n$-vertex graph $G$ and a tree-decomposition of $G$ of width~$t$, there is a $O(t^{O(t)} \cdot n)$-time
algorithm that computes the minimum $k$, if any, such that $G$ admits a (hole-free) $k$-map. In the positive case, the algorithm returns a certificate in the form of a~witness~of~$G$ within the same time complexity.
\end{theorem}

\noindent We remark that the problem of recognizing map graphs can be expressed by using MSO$_2$ logic. Thus the main positive result behind \cref{th:main} can be alternatively achieved by Courcelle's theorem~\cite{DBLP:journals/iandc/Courcelle90}. 
A formal proof is reported in the appendix. 
However, with this approach, the dependency of the time complexity on the treewidth is notoriously very high. As a matter of fact, Courcelle's theorem is generally used as a classification tool, while the design of an explicit ad-hoc algorithm  remains a challenging and valuable task~\cite{DBLP:books/sp/CyganFKLMPPS15}. 

\medskip To prove \cref{th:main}, we first solve the decision version of the problem. For a fixed $k$, we use a dynamic-programming approach, which can deal with different constraints on the desired witness. While we exploit such flexibility to check whether at most $k$ nations intersect at any point and whether holes can be avoided, other constraints could be plugged into the framework such as, for example, the outerplanarity of the witness (as in~\cite{DBLP:journals/disopt/MnichRS18}). In view of this versatility, future applications of our tools~may be expected.

\medskip\noindent\textbf{Proof strategy.} We exploit the characterization in~\cite{DBLP:journals/jacm/ChenGP02} and test for the existence of a suitable witness of the input graph. The crux of our technique is in the computation of suitable records  that represent equivalent witnesses and contain only vertices of a tree-decomposition bag. Each such record must carry enough information, in terms of embedding, so to allow testing whether it can be extended with a new vertex or merged with another witness. Moreover, we need to check whether any such witness yields a $k$-map and, if required, a hole-free one. To deal with the latter property, we provide a strengthening of the characterization in~\cite{DBLP:journals/jacm/ChenGP02}, which we believe to be of independent interest, that translates into maintaining suitable counters on the edges of our records. Additional checks on the desired witness can be plugged in the presented algorithmic framework,  provided that the records store enough information. One of the main difficulties is hence ``sketching'' irrelevant parts of the embedded graph without sacrificing too much  information. (A similar challenge is faced in the context of different planarity and beyond-planarity problems~\cite{DBLP:journals/jcss/GiacomoLM22,DBLP:conf/soda/JansenLS14,DBLP:journals/algorithmica/KociumakaP19}.) Also, when creating such sketches, multiple copies (potentially linearly many) of the same edge may appear, which we need to simplify to keep our records small. The formalization of such records then allows us to exploit a dynamic-programming approach on a tree-decomposition.

\medskip\noindent\textbf{Paper structure.} \cref{se:preliminaries} contains preliminary definitions.  \cref{se:basic} illustrates basic properties of map graphs that will be used throughout the paper. \cref{se:sketches} introduces the concept of ``sketching'' an embedding of a witness,  the key ingredient of the algorithmic framework, which we present in \cref{se:algorithm}.  \cref{se:conclusions} contains open problems raised by our work. 

\section{Preliminaries}\label{se:preliminaries}
We only consider finite, undirected, and simple graphs, although some procedures may produce non-simple graphs. In such a case the presence of self-loops or multiple edges will be clearly indicated. Let $G=(V,E)$ be a graph; for a vertex $v \in V$, we denote by~$N(v)$ the set of neighbors of $v$ in $G$, and by $\de{v}$  the \emph{degree} of $v$, i.e., the cardinality of $N(v)$.

\medskip\noindent\textbf{Embeddings.} 
A \emph{topological embedding}  of a graph $G$ on the sphere $\Sigma$ is a representation of $G$ on $\Sigma$ in which each vertex of $G$ is associated with a point and each edge of $G$ with a simple arc between its two endpoints in such a way that any two arcs intersect only at common endpoints. A topological embedding of $G$ subdivides the sphere into topologically connected regions, called \emph{faces}. If $G$ is connected, the  \emph{boundary} of a face $f$ is a closed walk, that is, a circular list of alternating vertices and edges; otherwise, the boundary of $f$ is a {\em set} of closed walks. Note that a cut-vertex of $G$ may appear multiple times in any such walk.  A topological embedding of $G$ uniquely defines a \emph{rotation system}, that is, a cyclic order of the edges around each vertex. If $G$ is connected, the boundary defining each face can be reconstructed from a rotation system; otherwise,  to reconstruct the boundary of every face $f$, we also need to know which connected components are incident to $f$. We call the incidence relationship between closed walks of different components and faces the \emph{position system} of $G$. A \emph{combinatorial embedding} of $G$ is an equivalence class of topological embeddings that define the same rotation and position systems. An \emph{embedded graph} $G$ is a graph along with a combinatorial embedding. A pair of parallel edges $e$ and $e'$ of $G$ with end-vertices $v$ and $w$ is  \emph{homotopic} if there is a  face of $G$ whose boundary consists of a single closed walk $\langle v,e,w,e' \rangle$.

\medskip\noindent\textbf{Tree-decompositions.} 
Let $(\mathcal{X},T)$ be a pair such that $\mathcal{X}=\{X_1,X_2,\dots,X_\ell\}$ is a collection of subsets of vertices of a graph $G$, called \emph{bags}, and $T$ is a tree whose nodes are in  one-to-one correspondence with the elements of $\mathcal X$. When this creates no ambiguity, $X_i$ will denote both a bag of $\mathcal{X}$ and the node of $T$ whose corresponding bag is $X_i$. The pair $(\mathcal{X},T)$ is a \emph{tree-decomposition} of $G$ if: (i) for every edge $(u,v)$ of $G$, there exists a bag $X_i$ that contains both $u$ and~$v$, and (ii) for every vertex $v$ of $G$, the set of nodes of $T$ whose bags contain $v$ induces a non-empty (connected)
subtree of $T$.
The \emph{width} of  $(\mathcal{X},T)$ is $\max_{i=1}^\ell {|X_i| - 1}$, while the \emph{treewidth} of $G$ is the minimum width over all tree-decompositions of $G$. For an $n$-vertex graph of treewidth $t$, a tree-decomposition of width $t$ can be found in FPT time~\cite{DBLP:journals/siamcomp/Bodlaender96}.

\begin{definition}\label{def:nice}
 A tree-decomposition $(\mathcal{X},T)$ of a graph $G$ is called \emph{nice} if $T$ is a rooted tree with the following properties~\cite{DBLP:journals/jal/BodlaenderK96}.
 \begin{enumerate}[(P.1)]
 \item Every node of $T$ has at most two children.
 \item If a node  $X_i$ of $T$ has two children whose bags are $X_j$ and $X_{j'}$, then $X_i=X_j=X_{j'}$. In this case, $X_i$ is a \emph{join bag}.
 \item If a node $X_i$ of $T$ has only one child~$X_j$, then $X_i \neq X_j$ and there exists a vertex $v \in G$ such that either $X_i = X_j \cup \{v\}$ or $X_i \cup \{v\} = X_j$. In the former case  $X_i$ is an \emph{introduce bag}, while in the latter case  $X_i$ is a \emph{forget bag}.
 \item If a node $X_i$ is a leaf of $T$, then $X_i$ contains exactly one vertex, and $X_i$ is a~\emph{leaf bag}.\end{enumerate} 
 \end{definition} 
 
Note that, given a tree-decomposition of width $t$, a nice tree-decomposition can be computed in $O(t \cdot n)$ time (see, e.g.,~\cite{DBLP:books/sp/Kloks94}).

\section{Basic Properties of Map Graphs and Their Witnesses}\label{se:basic}

The following statements  have already been discussed in the work by Chen et al.~\cite{DBLP:journals/jacm/ChenGP02}, even though in a weaker or different form. For completeness, we provide full proofs.

Let $G=(V,E)$ be a map graph and let $W = (V \cup I,A)$ be a witness of $G$, i.e., $W$ is a planar bipartite graph such that $W^2[V]=G$. A vertex $u \in I$ is an \emph{intersection vertex} of~$W$, while a vertex $v \in V$ is a \emph{real vertex} of~$W$. Also,~we let $n_V=|V|$, $n_I=|I|$,  and $n=n_V+n_I$.

\begin{property}\label{pr:kmapdegree}
A graph is a $k$-map graph if and only if it admits a witness such that the maximum degree of every intersection vertex is $k$.
\end{property}

\begin{property}\label{pr:connectivity}
A graph $G$ admits a map if and only if each of its biconnected components admits a map. Also, if $G$ admits a hole-free map, then $G$ is biconnected.
\end{property}

\begin{figure}[t]
\centering
\begin{subfigure}{.2\textwidth}
\centering
\includegraphics[page=1]{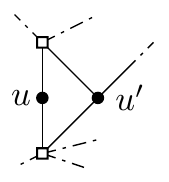}
\subcaption{\label{fig:redundant}}
\end{subfigure}\hfil
\begin{subfigure}{.2\textwidth}
\centering
\includegraphics[page=2]{figs/redundant-twinpair}
\subcaption{\label{fig:twinpair}}
\end{subfigure}
\caption{(a) Inessential intersection vertices, and (b) a twin-pair.}
\end{figure} 
\noindent Let $W = (V \cup I,A)$ be an embedded witness (i.e., with a prescribed combinatorial embedding). An intersection vertex $u \in I$ is \emph{inessential} if $\de{u}=2$ and there exists $u' \in I$ such that $N(u) \subset N(u')$; see \cref{fig:redundant}.  Furthermore, a pair of intersection vertices $u_1,u_2 \in I$ is a \emph{twin-pair} if $N(u_1)=N(u_2)=\{v,w\}$, for some $v,w \in V$, and $W$ contains a face whose boundary consists of a single closed walk with exactly four edges with end-vertices  $v,u_1,w,u_2$; see \cref{fig:twinpair}. Note that removing an inessential vertex or one vertex of a twin-pair from $W$~does~not~modify~$W^2[V]$.

\begin{definition}\label{def:compact}
An embedded witness of a map graph is \emph{\compact} if it contains neither inessential intersection vertices nor twin-pairs. 
\end{definition}

We remark that a \compact witness is not necessarily minimal, i.e., it may contain intersection vertices of degree greater than $2$ whose removal does not modify its half-square; see also~\cite{DBLP:journals/jacm/ChenGP02}.  However, in our setting, removing further information from a witness would have an impact on the proof of \cref{th:holefree}  and  on the recognition algorithm (\cref{se:algorithm}). 

 The next lemma shows that focusing on \compact witnesses is not restrictive.

\begin{restatable}{lemma}{thcharcompact}\label{th:char-compact}
A graph $G=(V,E)$ is a map graph if and only if it admits a \compact witness. Also, $G$ is a $k$-map graph if and only if it admits a \compact witness whose intersection vertices have degree at most $k$.
\end{restatable}
\begin{proof}
For the first part of the statement, recall that a graph admits a map if and only if it admits a witness~\cite{DBLP:conf/stoc/ChenGP98,DBLP:journals/jacm/ChenGP02}. Thus, if $G$ has a \compact witness, it is a map graph.
For the other direction, suppose that~$G$ admits a map. Let $\hat{W}$ be any embedded witness of $G$. 
Let $W$ be the embedded graph obtained from~$\hat{W}$ by removing all inessential intersection vertices and by iteratively removing one intersection vertex for each twin-pair. Since we only removed degree-$2$ intersection vertices whose neighbors are already incident to a common intersection vertex, it holds that $W^2[V]=\hat{W}^2[V]$. Thus, since $\hat{W}^2[V]=G$, it holds $W^2[V]=G$.
For the second part of the statement, recall that $G$ admits a $k$-map if and only if it has a witness whose intersection vertices have degree at most $k$, by \cref{pr:kmapdegree}. Moreover, we have seen before that, for every witness $\hat{W}$, there exists (at least) one \compact witness~$W$ obtained by fixing a combinatorial embedding of $\hat{W}$ and by possibly removing intersection vertices of degree 2 that are either inessential or part of a twin-pair. 
Since $\hat{W}$ is bipartite, this implies that any intersection vertex that belongs to both $\hat{W}$ and $W$ has the same degree in the two graphs. 
\end{proof}

Given a graph $G$ such that $n_V \ge 3$ and a map $\mathcal M$ of $G$, the \emph{order} of a point $p \in \cal M$, denoted~by~$ord(p)$, is equal to the number of nations and holes whose boundary contains $p$. Let $W = (V \cup I,A)$ be the bipartite embedded graph computed with the \emph{\construction}, defined as follows. For ease of description, we define $W$ by constructing a topological embedding of it; refer to \cref{fig:map}. In particular, the witness of \cref{fig:map-2} is \compact and constructed with the described procedure (which again follows the lines of the work in~\cite{DBLP:journals/jacm/ChenGP02}). For each nation~$\mfunc{v}$, we place the real vertex $v$ in its interior. For each point $p$ such that $ord(p) \ge 3$, we add an intersection vertex $u_p$ to~$I$ and place it at point $p$. We connect each real vertex $v$ to the intersection vertices that lie on the boundary of $\mfunc{v}$, by drawing  crossing-free simple arcs inside $\mfunc{v}$. Note that, for each intersection vertex $u_p$, it holds $\textrm{deg}(u_p)=ord(p)-h(p)$, where $h(p)$ is the number of holes in $\mathcal M$ whose boundary contains~$p$. Finally, we remove inessential intersection vertices and, iteratively, a vertex for each twin-pair. For instance, in \cref{fig:map-1} the nations colored light-yellow and light-green share two order-$3$ points that would give rise to inessential intersection vertices, which are indeed~not~reported~in~\cref{fig:map-2}.

\begin{lemma}\label{le:compact}
Let $W = (V \cup I,A)$ be the embedded graph obtained from the map $\mathcal M$ of $G$ by means of the \construction. Then $W$ is a \compact witness of $G$.
\end{lemma}
\begin{proof}
The fact that the \construction defines a topological embedding of $W$, and in particular that each arc is simple and no two arcs intersect at an interior point, follows by construction. Moreover, we explicitly removed inessential intersection vertices and twin-pairs, if any.  So, it remains to prove that $W^2[V]=G$. By construction, for each edge of $W^2[V]$, there is an edge in $G$. Also, any edge of $G$ is represented by at least one point of $\mathcal M$ whose order is at least $2$. Since we created an intersection vertex for each point of order greater than $2$, it remains to argue about points of order exactly $2$. Any such a point is an interior point of a simple arc along which two nations $\mfunc{v}$ and $\mfunc{w}$ touch. The two endpoints of this arc must have order at least $3$, which implies that edge $(v,w)$ exists in $W^2[V]$.
\end{proof}

\noindent In~\cite{DBLP:journals/jacm/ChenGP02}, it is observed (without a formal argument) that a map graph is hole-free if and only if it admits a witness whose faces have 4 or 6 edges each. The next characterization improves over this observation and hence can be of independent interest. A connected embedded graph is a \emph{quadrangulation} if each face boundary consists of a single closed walk with 4 edges.

\begin{restatable}{theorem}{thholefree}\label{th:holefree}
A graph is a hole-free map graph if and only if it admits a \compact witness that is a biconnected quadrangulation.
\end{restatable}
\begin{proof}
\begin{figure}[t]
    \centering
    \begin{subfigure}{.22\textwidth}
    \includegraphics[page=1,width=\textwidth]{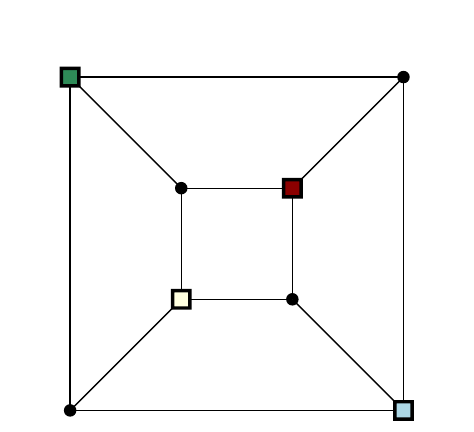}
    \subcaption{$W$\label{fig:holefree-a}}
    \end{subfigure}\hfil
    \begin{subfigure}{.22\textwidth}
    \includegraphics[page=5,width=\textwidth]{figs/sphere}
    \subcaption{$W^*$\label{fig:holefree-b}}
    \end{subfigure}\hfil
    \begin{subfigure}{.22\textwidth}
    \includegraphics[page=6,width=\textwidth]{figs/sphere}
    \subcaption{$\cal M$\label{fig:holefree-c}}
    \end{subfigure}
    \caption{Illustration for the proof of \cref{th:holefree}.\label{fig:holefree}}
\end{figure}
($\Leftarrow$) Refer to \cref{fig:holefree}. If a graph $G$ admits a \compact witness $W$, then  by \cref{th:char-compact} $G$ is a map graph. Thus we only need to show that $W$ yields a map that is hole-free, by exploiting the assumption that $W$ is a biconnected quadrangulation. Let $W^*$ be the embedded graph defined as follows: We add a dummy vertex inside each face of $W$ and connect it to all vertices on the boundary of the face. Since $W$ is biconnected, each face boundary is a simple cycle, and therefore $W^*$ is an embedded triangulation, i.e., each face contains three edges on its boundary. Let $\Gamma$ be a topological embedding on the sphere~of~$W^*$. It follows that each triangular face of $\Gamma$ is incident to exactly one real vertex of $W$.  Also, for each real vertex $v$ of $W$, the union of the triangular faces incident to $v$ defines a region $R_v$ that contains $v$ in its interior. The latter property of $\Gamma$ allows us to construct a map~$\cal M$ of $G$ by setting~$\mfunc{v}$ to be equal to the closure of~$R_v$, for each real vertex $v$ of $W$. The fact that $\cal M$ is a hole-free map follows by construction. Namely, all points of the sphere are covered by nations, hence there are no holes. Also, the points of $\cal M$ of order at least $3$ are in a one-to-one correspondence with the intersection vertices of $W$, and any other order-$2$ point of $\cal M$ lies along a simple arc connecting two points of higher order.

($\Rightarrow$) Let $\cal M$ be a hole-free map of a graph $G$. By \cref{pr:connectivity}, $G$ is biconnected.
Let  $W$ be the \compact witness of $G$ computed by the \construction  from $\cal M$; for instance, the \compact witness in \cref{fig:holefree-a} is constructed from the map in \cref{fig:holefree-c} by using the \construction. Since $G$ is biconnected, $W$ is (at least) connected, and thus the boundary of any face of $W$ consists of a single closed walk. In the following we prove that $W$ is a quadrangulation. This, together with the fact that $W$ is simple, implies biconnectivity. Since $W$ is connected, simple and bipartite, the boundary of each of its faces consists of a single closed walk containing at least four edges. Assume for a contradiction that $W$ contains a face $f$ with more than four edges on the closed walk $\pi$ defining its boundary. Then, since $W$ is bipartite, $\pi$ contains at least six edges. Consider any intersection vertex $u$ on $\pi$.  Let $v$ and $w$ be the two (real) vertices that precede and follow $u$ along $\pi$, respectively. We distinguish two cases based on whether $u$ is or is not the only intersection vertex in $\pi$ that is adjacent to both $v$ and $w$.

Suppose first that $u$ is the only intersection vertex in $\pi$ that is adjacent to both $v$ and $w$. By construction, $u$ has been placed on a point $p \in \cal M$ such that $ord(p) \ge 3$. Since $\cal M$ is hole-free, $p$ is the endpoint of a simple arc $a$ that forms a shared boundary between $\mfunc{v}$ and $\mfunc{w}$. Let $p'$ be the other endpoint of $a$. Since $ord(p') \ge 3$ and $\cal M$ is hole-free, again $W$ contains an intersection vertex $u'$ that has been placed on $p'$. Either $u'$ belongs to $\pi$ or not. In the first case, we contradict the fact that $u$ is the only intersection vertex of $\pi$ adjacent to both $v$ and $w$; see \cref{fi:quad-1} for an illustration. In the second case, $p$ is on the boundary of a hole, which  contradicts the fact that $\cal M$ is hole-free; see \cref{fi:quad-2} for an illustration. 

Suppose now that there is another intersection vertex $u'$ in $\pi$ adjacent to both $v$ and $w$; refer to \cref{fi:quad-3} for an illustration. Since $u$ and $u'$ are both adjacent to $v$ and $w$ and both belong to $\pi$, in order for $\pi$ to contain six (or more) edges, at least one of $v,u,w,u'$ occurs more than once along $\pi$, and between its (at least) two occurrences, there must be a real vertex $w'$. Note that, by definition of $v$ and $w$, such vertex occurring multiple times along $\pi$ cannot be $u$. We assume that such vertex is $u'$, the remaining cases can be handled with symmetric arguments. 
\begin{figure}[h]
\centering
\begin{subfigure}{.32\textwidth}
\centering
\includegraphics[page=1]{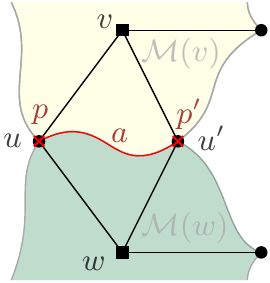}
\subcaption{\label{fi:quad-1}}
\end{subfigure}\hfil
\begin{subfigure}{.32\textwidth}
\centering
\includegraphics[page=2]{figs/compactquad}
\subcaption{\label{fi:quad-2}}
\end{subfigure}\hfil
\begin{subfigure}{.32\textwidth}
\centering
\includegraphics[page=3]{figs/compactquad}
\subcaption{\label{fi:quad-3}}
\end{subfigure}
\caption{Illustrations for the proof of \cref{th:holefree}.\label{fi:quad}}
\end{figure}Consider a traversal of $\pi$ that visits $v,u,w$ in this order. Up to a renaming of the vertices, we can assume that $w'$ is the vertex that precedes the last occurrence of $u'$ in this traversal. Let $q$ be the point of $\cal M$ where $u'$ has been placed on. By construction, $q$ is the endpoint of a simple arc $b$ that forms a shared boundary between $\mfunc{w'}$ and $\mfunc{v}$. Let $q'$ be the other endpoint of $b$. Since $ord(q') \ge 3$ and $\cal M$ is hole-free, again $W$ contains an intersection vertex $u''$ that has been placed on $q'$ and that belongs to $\pi$ by construction. Since $u''$ is adjacent to both $v$ and $w'$, we get a contradiction to the fact that the last occurrence of $u'$ before $v$ is encountered after $w'$ in the  above traversal.
\end{proof}

\begin{restatable}{lemma}{lesize}\label{le:size}
A (hole-free) map graph $G$ admits a \compact witness with $n \le 6n_V-10$ (respectively, $n \le 3n_V-4$) vertices.   
\end{restatable}
\begin{proof}
Suppose first that $G$ is hole-free. Let $W=(V \cup I,A)$ be a \compact witness of $G$ that is a quadrangulation, which exists by \cref{th:holefree}.  We start with the following claim.

\begin{claim}\label{cl:no-deg2}
 $\forall u \in I$, it holds $\de{u}>2$.
\end{claim}
\begin{proof}[Proof of the claim]
 Suppose, for a contradiction, that $W$ contains an intersection vertex $u$ such that $\de{u}=2$ and let $f$ be any face of $W$ that contains $u$ on its boundary. Let $u'$ be the other intersection vertex on the boundary of $f$, and observe that $N(u) \subseteq N(u')$. If $N(u)=N(u')$, then $u$ and $u'$ form a twin-pair. Otherwise $u$ is inessential. Both cases contradict the fact that $W$ is \compact.
\end{proof}

Since $W$ is crossing-free and bipartite (because $W$ is a witness), it holds $|E| =\sum_{u \in I}\textrm{deg}(u) \le 2(n_V+n_I)-4$. By \cref{cl:no-deg2}, we have $\sum_{u \in I}\textrm{deg}(u) \ge 3n_I$. Putting all together, we have $2(n_V+n_I)-4 \ge \sum_{v \in I}\textrm{deg}(u) \ge 3n_I$. Consequently, $n_I \le 2n_V-4$ and thus $n = n_V+n_I \le 3n_V-4$.

Suppose now that $G$ is not hole-free. Let $W$ be any \compact witness of $G$ and let $W'$ be the graph obtained by removing all degree-$2$ intersection vertices from $W$. 
Since $W'$ is crossing-free and bipartite and since its intersection vertices have degree at least 3, we can conclude as above that $W'$ has at most $3n_V-4$ vertices. We now claim that $W$ has no more than $3n_V-6$ degree-$2$ intersection vertices, which concludes the proof, since these are exactly the vertices in $W \setminus W'$. To prove the claim, replace each degree-$2$ intersection vertex of $W$ with an edge connecting its two neighbors. Since $W$ is crossing-free, the resulting graph $W^*$ is also crossing-free. Also, each intersection vertex in $W^*$ has degree greater than $2$. Since $W$ contains no twin-pairs, $W^*$ contains no pairs of homotopic parallel edges. Thus, Euler's formula for planar graphs still applies and therefore $W^*$ contains at most $3n_V-6$~edges. Since each edge of $W^*$ corresponds to at most one degree-$2$ intersection vertex of $W$, the claim follows.
\end{proof}

Based on \cref{le:size}, we can make the following remark.

\begin{remark}\label{re:size}
Without loss of generality, we assume in the following that any \compact witness $W$ of $G$ has $n \le 3n_V-4$ vertices if $G$ is hole-free, or  $n \le 6n_V-10$ vertices otherwise.
\end{remark}

\section{Embedding Sketches}\label{se:sketches}
Let $G$ be an input graph. \cref{pr:connectivity} allows us to assume that $G$ is~biconnected, and thus every witness of $G$, if any, is connected. Also, by \cref{th:char-compact}, it suffices to consider \compact witnesses.

Let $(\mathcal{X}, T)$ be a nice tree-decomposition of $G$ of width $t=\omega-1$, i.e., each bag contains at most~$\omega$ vertices. 
Given a bag $X \in \cal X$, we denote by $T_X$ the subtree of $T$ rooted at $X$, and by $G_X=(V_X,E_X)$ the subgraph of $G$ induced by all the vertices in all bags~of~$T_X$.  Let $W_X=(V_X \cup I_X, A_X)$ be a \compact witness of $G_X$ (in particular,  $W_X^2[V_X]=G_X$). Note that, although $G$ is connected, $G_X$ may have multiple connected components. However, since $G$ is connected, each connected component of $G_X$ must contain at least one vertex~of~$X$. Moreover, for each connected component $C$ of $G_X$, there is a connected component $C'$ of $W_X$ such that $C'$ is a witness of $C$. A vertex of $W_X$ is an \emph{anchor vertex} if it is either a real vertex of $X$ or an intersection vertex whose neighbors in $W_X$~all~belong~to~$X$. Observe that if an intersection vertex $u$ has a neighbor $v$ in $V_X \setminus X$, then no real vertex in $V \setminus V_X$ is adjacent to $v$, and therefore there is no way to add further edges to $u$ without creating a false adjacency involving $v$.

\begin{figure}[t]
\centering
\begin{subfigure}{.3\textwidth}
\centering
\includegraphics[page=3,width=\textwidth]{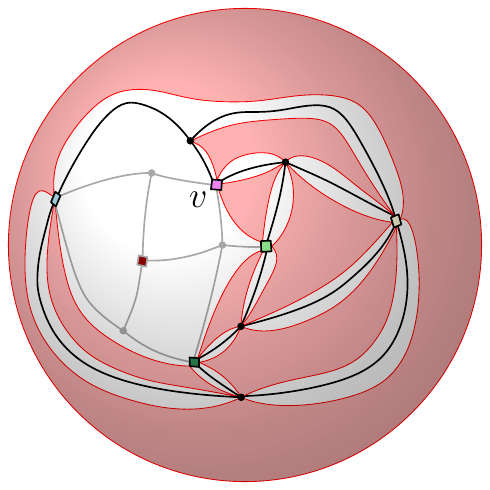}
\subcaption{$S(W,X')$\label{fig:sketched-a}}
\end{subfigure}\hfil
\begin{subfigure}{.3\textwidth}
\centering
\includegraphics[page=4,width=\textwidth]{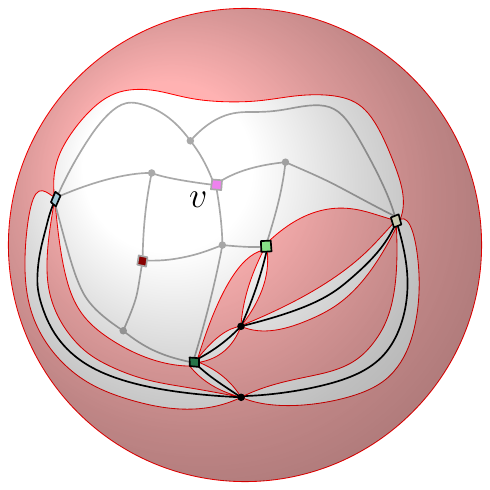}
\subcaption{$S(W,X)$, $X \cup \{v\} = X'$\label{fig:sketched-b}}
\end{subfigure}\hfil
\caption{(a) A \sketch $S(W,X')$ computed from the witness $W$ of \cref{fig:map} with respect to a bag $X'$ ($V_{X'}=V$). The anchor vertices of $X'$ are opaque, while the non-anchor vertices are faded. The active boundaries are red and the background of the active faces is light red. (b) A \sketch  $S(W,X)$, where $X \cup\{v\}=X'$ computed from $S(W,X')$ by applying the deletion operation  (\cref{se:algorithm}).\label{fig:sketched}}
\end{figure} 

We will exploit anchor vertices to reduce the size of $W_X$ from $O(|V_X|)$ to $O(\omega)$, by ``sketching'' parts of the embedding that are not relevant\footnote{In the database and data engineering fields, sketching algorithms form a powerful toolkit to compress data in a way that supports answering various queries~\cite{DBLP:journals/queue/Cormode17}. Our idea of sketching has some similarities with this concept but serves a different purpose.}. The idea of sketching an embedded graph is inspired by a previous work about orthogonal planarity~\cite{DBLP:journals/jcss/GiacomoLM22}; applying such idea to our problem requires the development of several new tools and concepts, described in  the remainder of this section (and partly in \cref{se:preliminaries}).  
A face $f$ of $W_X$ is \emph{active} either if its boundary contains only one vertex $v$ (which implies  $W_X=(\{v\},\emptyset)$) and $v$ is an anchor vertex, or if its boundary contains more vertices among which there are at least two anchor vertices; refer to \cref{fig:sketched-a}.   The  \emph{active boundary} of $f$ (red in \cref{fig:sketched-a}) is obtained by shortcutting all non-anchor vertices of $f$, where the \emph{shortcut operation} is defined as follows. 
 For a closed walk $\pi$ and a vertex $v$ in $\pi$, shortcutting $v$ consists of removing each occurrence of $v$ (if more than one), together with the edge $(u,v)$ that precedes it in $\pi$, and the edge $(v,u')$ that follows it in $\pi$, and of adding the edge $(u,u')$ between $u$ and $u'$ in $\pi$.
\begin{figure}[tb]
\centering
\includegraphics[page=1]{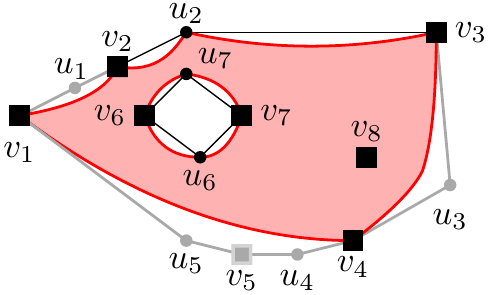}
\caption{An active boundary (red) made of three closed walks (edges are omitted): $\langle v_1, v_2, u_2, v_3, v_4,  v_1 \rangle$, $\langle u_6, v_6, u_7, v_7 \rangle$, $\langle v_8 \rangle$; vertices $u_1,u_3,u_4,v_5,u_5$ have been shortcut.\label{fig:disconnected}}
\end{figure}
\cref{fig:disconnected} illustrates a single face $f$ and the corresponding active boundary. 
 The \emph{embedding \sketch} (for short the \emph{\sketch}) \emph{of $W_X$ with respect to $X$}  is the embedded graph $S(W_X,X)$ formed by all the vertices and edges that belong to the active boundaries of $W_X$. For each active boundary $B_f$ of an active face $f$ of $W_X$, $S(W_X,X)$ has an \emph{active face}  $f^*$ (light red in \cref{fig:sketched-a,fig:disconnected}). Note that $S(W_X,X)$ also has faces that are not active (white in \cref{fig:sketched-a,fig:disconnected}). Also, the position system of $W_X$ yields a position system for $S(W_X,X)$, since if two closed walks of distinct components of $W_X$ were incident to the same active face $f$, then the two corresponding closed walks of $S(W_X,X)$ are also incident to the same active face $f^*$. However, $S(W_X,X)$ may not be bipartite any longer (as in \cref{fig:disconnected}) and it may contain multiple edges (but no self-loops). It is worth noting that the embedding \sketch of $W_X$ can be defined with respect to any bag $X'$ as long as $V_{X'}=V_X$ (see \cref{fig:sketched-a}). 
 
 We now further refine $S(W_X,X)$ to avoid active boundaries that are not useful for our purposes. Namely, an active boundary is \emph{non-extensible} if it consists of two homotopic parallel edges.  Given a witness $W$ of $G$, the \emph{restriction} of $W$ to $G_X$ is the \compact witness  $W[G_X]$ of $G_X$ obtained from $W$ by removing all the real vertices not in $G_X$, all the intersection vertices that are isolated (due to the removal of some real vertices) or inessential, as well as a vertex for each twin-pair until the graph contains none of them. \mbox{The next lemmas allow us to bound the size~of~a~\sketch.}

\begin{lemma}\label{le:no-parallel}
If $G$ is a map graph, then it admits a \compact witness $W$ with the following property. If $S(W[G_X],X)$ contains $h>1$ non-extensible active boundaries that~share the same pair of end-vertices, then the vertices of $W$ lie in at most one of these $h$ active~boundaries.
\end{lemma}
\begin{proof}
\begin{figure}[htb]
\centering
\begin{subfigure}{.35\textwidth}
\centering
\includegraphics[page=1,width=\textwidth]{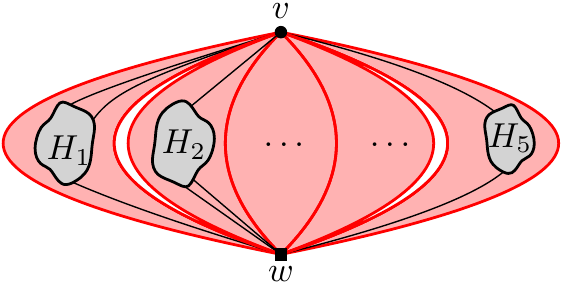}
\subcaption{$\hat{W}$\label{fi:no-parallel-1}}
\end{subfigure}\hfil
\begin{subfigure}{.35\textwidth}
\centering
\includegraphics[page=2,width=\textwidth]{figs/parallel}
\subcaption{$W$\label{fi:no-parallel-2}}
\end{subfigure}
\caption{Illustrations for the proof of \cref{le:no-parallel}. Modifying the rotation system of $\hat{W}$ such that each $H_i$ lies in $B_1$ and all other non-extensible active boundaries become empty.\label{fi:no-parallel}}
\end{figure}
Refer to \cref{fi:no-parallel}.
Let $\hat{W}$ be a \compact witness of $G$, and suppose  $S(\hat{W}[G_X],X)$ contains $h>1$ non-extensible active boundaries $B_1,B_2,\dots,B_h$ with common end-vertices $v,w$. Let $H_i$ be the subgraph of $W$ that lies inside $B_i$ (if any), for $1 \le i \le h$. Since each $B_i$ consists of two parallel edges, $v$ and $w$ separate $H_i$ and $S(W[G_X],X) \setminus H_i$. We obtain a new \compact witness $W$ of $G$ by~modifying~the~rotation~system~of~$\hat{W}$~so~that~each~$H_i$~lies~inside~$B_1$.\end{proof}

\begin{remark}\label{re:witness}
By \cref{le:no-parallel}, we assume in the following that for any \compact witness $W$ of~$G$ such that, for some $X \in \mathcal{X}$, the \sketch $S(W[G_X],X)$ contains $h>1$ non-extensible active boundaries, the vertices of $W$ lie in at most one of such active boundaries. Therefore, in $S(W[G_X],X)$,  we keep only one of the corresponding $h$ pairs of homotopic parallel edges.
\end{remark} 

\begin{lemma}\label{le:sketched-size}
A \sketch $S(W_X,X)$ contains $O(\omega)$ vertices and edges.
\end{lemma}
\begin{proof}
With a similar argument as in the proof of \cref{le:size} we can show that, in $W_X$, each real vertex in $X$ is adjacent to $O(\omega)$ intersection vertices that are anchor vertices. Therefore, $S(W_X,X)$ contains $O(\omega)$ vertices in total.  Concerning the number of edges, since $S(W_X,X)$ is embedded on the sphere, it contains $O(\omega)$ edges such that each pair of edges is either non-parallel or non-homotopic parallel. In addition, since each of these edges participates in at most one homotopic pair by \cref{re:witness}, it follows that $S(W_X,X)$ contains $O(\omega)$ edges.
\end{proof}

\noindent We now exploit the concept of \sketch to define an equivalence relation among witnesses.

\begin{definition}\label{def:x-equiv}
	Two \compact witnesses $W_X$ and $W'_X$ of $G_X$ are \emph{$X$-equivalent}  if they have the same \sketch with respect to $X$, i.e., $S(W_X,X)=S(W'_X,X)$.
\end{definition}

\noindent The next lemma deals with the size of the quotient of such a relation.

\begin{lemma}\label{lem:number}
The $X$-equivalence relation yields $\omega^{O(\omega)}$ classes for the \compact witnesses of $G_X$.
\end{lemma}

\begin{proof}
Let $n_1$ be the number of possible (abstract) graphs that can be obtained from the real vertices of $X$ and all possible sets of intersection vertices. For each such graph, let $n_2$ be the maximum number of possible rotation and position systems that it can have. It follows that the number of $X$-equivalent classes is upper bounded by the product of $n_1$ and $n_2$.

Given the set $X$ of real vertices and a \compact witness $W_X$ of $G_X$,  any \sketch $S(W_X,X)$ contains $O(\omega)$ intersection vertices, as otherwise $W_X$ would contain inessential intersection vertices or twin-pairs. Since each intersection vertex is adjacent to a set of at most $\omega$ real vertices, we can bound the number $n_{\textrm{int}}$ of possible sets of intersection vertices by $a \cdot \sum_{i=2}^\omega {\omega \choose i} < a \cdot 2^\omega$, where $a$ is the maximum number of intersection vertices in any \sketch that have the same set of neighbors. Since $a \in O(\omega)$, we have that $n_{\textrm{int}} \in 2^{O(\omega)}$. Let $I_X$ be one of the $n_{\textrm{int}}$ possible sets of intersection vertices. The number $n_{\textrm{abs}}$ of distinct abstract graphs with vertex set $X \cup I_X$ can be upper bounded by the number of possible neighborhoods of a real vertex combined for all real vertices, that is 

$$n_{\textrm{abs}} \le \prod_{v \in X} \omega ^{\textrm{deg}(v)} = \omega^{\sum_{v \in X}{\textrm{deg}(v)}} \le \omega^{O(\omega)}$$

\noindent holds, which yields $n_1 \le n_{\textrm{int}} \cdot n_{\textrm{abs}} \in \omega^{O(\omega)}$.
	
For a fixed graph $S$,  the number of possible rotation systems $n_{\textrm{rot}}$ is upper bounded by the number of possible permutations of edges around each vertex.  Thus we have 

$$n_{\textrm{rot}} \le \prod_{v \in S} {\textrm{deg}(v)!} < \prod_{v \in S} {\textrm{deg}(v)^{\textrm{deg}(v)}} \le \omega^{O\left(\sum_{v \in S}{\textrm{deg}(v)}\right)} \le \omega^{O(\omega)}.$$ 

\noindent Each rotation system of $S$ fixes the closed walk of each face of each connected component of $S$. Since $S$ contains, over all its connected components, at most $\omega$ closed walks (at most one for each real vertex in $X$) and hence at most $\omega$ faces, for the number $n_{\textrm{pos}}$ of possible position systems it holds $n_{\textrm{pos}} \le \omega^\omega$. Therefore we have $n_2 \leq n_{\textrm{rot}} \cdot n_{\textrm{pos}} \in \omega^{O(\omega)}$, which yields $n_1 \cdot n_2 \in \omega^{O(\omega)}$, as desired.
\end{proof}

\section{Algorithmic Framework}\label{se:algorithm}
Let $G=(V,E)$ be an input graph, let $k$ be an integer, and let $(\mathcal{X}, T)$ be a nice tree-decomposition of $G$ of width $t=\omega-1$. We present an algorithmic framework to test whether $G$ is a $k$-map graph or a hole-free $k$-map graph. Namely, we traverse $T$ bottom-up and equip each bag $X \in \mathcal{X}$ with a suitably defined set of \sketches, called \emph{record $R_X$}. The framework can be tailored by imposing different properties for the records. The next three properties are rather general; the first two  are useful to prove the correctness of our approach, as shown in \cref{th:correct}, whereas the third property comes into play when dealing with the efficiency of the approach, and in particular in \cref{le:rx}.

\begin{definition}\label{def:feasible}
The record $R_X$ is \emph{\feasible} if the following properties hold: 
\begin{enumerate}[{F}1]
    \item \label{f1} For every \compact witness $W_X$ of $G_X$, $R_X$ contains its \sketch $S(W_X,X)$. 
    \item \label{f2} For every entry $r \in R_X$, there is a \compact witness $W_X$ of $G_X$ such~that~$r=S(W_X,X)$.
    \item \label{f3} $R_X$ contains no duplicates.
\end{enumerate}
\end{definition}

\begin{lemma}\label{le:rx}
For every $X \in \mathcal{X}$, if $R_X$ is \feasible,  it contains $\omega^{O(\omega)}$ entries, each of size $O(\omega)$.
\end{lemma}
\begin{proof}
By F\ref{f1}--F\ref{f3}, the entries of $R_X$ are all and only the possible \sketches of $W_X$ and are all distinct. Hence, $|R_X| \in \omega^{O(\omega)}$ by \cref{lem:number}. Each \sketch has size $O(\omega)$ by \cref{le:sketched-size}.
\end{proof}

We now describe the additional properties that we incorporate in the framework. In order to verify that $G$ admits a  $k$-map we exploit \cref{pr:kmapdegree}, which translates into verifying that, for each \sketch, the degree of any intersection vertex is at most $k$. 

\begin{definition}\label{def:k-feasible}
A record $R_X$ is \emph{$k$-map \feasible} if it is \feasible and it contains a non-empty subset $R^*_X \subseteq R_X$, called \emph{subrecord}, for which the following additional property holds: 
\begin{enumerate}[{F}1]\setcounter{enumi}{3}
    \item \label{f4}  For every entry $r \in R_X$, it holds $r \in R^*_X$ if and only if $r$ contains no intersection vertex $u$ with  $\de{u}>k$.
\end{enumerate}
\end{definition}

\noindent It is worth observing that, since an intersection vertex of degree $k$ implies the existence of a clique of size~$k$ in the input graph $G$, property F\ref{f4} is trivially verified when $k \ge \omega$. On the other hand, the size of the largest clique of a $k$-map graph is $\lfloor 3k/2 \rfloor$ (see, e.g.,~\cite{DBLP:journals/jacm/ChenGP02}). 

\smallskip

To check whether $G$ has a hole-free $k$-map,  we exploit \cref{th:holefree}. Namely, consider a \sketch $S(W_X,X)$ and an active boundary $B_f$ of $S(W_X,X)$. Let $f$ be the active face of $W_X$ corresponding to $B_f$. Note that any edge~$e$ that is part of $B_f$ represents a subsequence~$\pi_e$ of a closed walk $\pi$ in the boundary of $f$. Therefore, to control the number of edges on the boundary of each face~of~$W_X$, for every edge $e$ that is part of an active boundary of $S(W_X,X)$ we also store a counter $c(e) \ge 1$ , which represents the number of edges~in~$\pi_e$. If there is an edge $e$ such that $c(e) > 4$, then $G$ does not admit a \compact witness $W$ that is a quadrangulation and such that $W_X=W[G_X]$; hence we can avoid storing counters greater than $4$. Moreover, for any face $f$ of a compact witness $W$ of $G$, we know there exist two bags $\hat{X'}$ and $\hat{X}$ in $T$ such that $\hat{X'}$ is the child of $\hat{X}$, $\hat{X}$ is a forget bag, the active boundary representing $f$ in $\hat{X'}$  has more than one anchor vertex, while the one in $\hat{X}$ has only one anchor vertex (and hence is not part of $S(W_{\hat{X}},\hat{X})$). We call such an active boundary \emph{complete}~in~$\hat{X'}$, as it will not be modified anymore by the algorithm. As such, for each complete active boundary, the sum of the counters of its edges in $S(W_{\hat{X'}},\hat{X'})$ must be exactly~4, otherwise~$G$ does not admit a \compact witness $W$ that is a biconnected  quadrangulation such~that~$W_{\hat{X}}=W[G_{\hat{X}}]$.

\begin{definition}\label{def:hf-feasible}
A record $R_X$ is \emph{hole-free \feasible} if it is \feasible and it contains a non-empty subset $R^\circ_X \subseteq R_X$, called \emph{subrecord}, for which the following additional property holds:  
\begin{enumerate}[{F}1]\setcounter{enumi}{4}
    \item \label{f5} For every entry $r \in R_X$, it holds $r \in R^\circ_X$ if and only if $r$ contains no intersection vertex~$u$ with  $\de{u}>k$ and each complete active boundary of $r$ (if any) is such that its edge counters sum~up~to~4.
\end{enumerate}
\end{definition}

Each leaf bag contains only one vertex~$v$, thus its record consists of one \sketch with only one active face whose active boundary is $\langle v \rangle$. Such a  record can be computed in $O(1)$ time and it is trivially \feasible. Also, it is hole-free (and hence $k$-map) \feasible, as its unique active boundary is not complete. The next three operations are performed on a non-leaf bag $X$ of $T$, based on the type of $X$, to compute a $k$-map or hole-free \feasible record $R_X$, if any.

\medskip\noindent\textbf{Deletion operation.} 
Let $X$ be a forget bag whose child $X'$ in $T$ has a $k$-map (hole-free) \feasible record $R_{X'}$. Let $v$ be the vertex forgotten by $X$. We generate $R_X$ from $R_{X'}$ as follows. 

For a fixed \sketch $S(W_{X'},X')$ of $R_{X'}$, let $N_I(v) \subseteq N(v)$ be the set of intersection vertices adjacent to $v$ in $S(W_{X'},X')$. Since $v$ is forgotten by $X$, all its neighbors have already been processed, thus no vertex in $N_I(v)$ can  connect vertices that will be introduced by bags visited after $X$. Therefore, for every vertex $y \in N_I(v) \cup \{v\}$ and for every  \sketch $S(W_{X'},X')$ of $R_{X'}$, we apply a \emph{deletion operation}, which consists of updating each active boundary~$B_f$ of $S(W_{X'},X')$ containing $y$; see \cref{fig:sketched-b}. Namely, let $B_f$ be one of these active boundaries, we distinguish two cases based on whether $B_f$ contains only $y$ or it contains further vertices. 
Let $\pi_y$ be the closed walk of $B_f$ that contains all occurrences of $y$ (there might be more than one). If $B_f$ contains only $y$, we remove $\pi_y$ (and hence the whole active boundary $B_f$) from $S(W_{X'},X')$. If $B_f$ contains further vertices, we shortcut every occurrence of $y$ in $\pi_y$.  Also for each edge~$e$ introduced to shortcut $y$ such that~$e$ replaces edges $e_1$ and $e_2$ of $\pi_y$, we set $c(e)=c(e_1)+c(e_2)$.  Observe that, if $y$ has only one neighbor~$u$~in~$\pi_y$, this procedure creates a self-loop at~$u$, which we remove. If this procedure generates more than one pair of homotopic parallel edges with the same pair of end-vertices, then we keep only one such pair. Once all active boundaries have been updated, the resulting embedded graph is stored in~$R_X$. After each \sketch of $R_{X'}$ has been processed, we might have produced the same embedded graph for $R_X$ from two distinct \sketches of $R_{X'}$; in this case we keep only one copy.

\medskip\noindent\textbf{Addition operation.} 
Let $X$ be an introduce bag whose child $X'$ in $T$ has a $k$-map (hole-free) \feasible record $R_{X'}$.
Let $v$ be the vertex introduced by $X$ and $N_X(v) \subseteq N(v)$ be the set of vertices that are neighbors of $v$ and belong to $X$. We generate $R_X$ from $R_{X'}$ with the following \emph{addition operation}. For each \sketch $S(W_{X'},X')$ of $R_{X'}$, the high-level idea is to exhaustively generate all possible embedded graphs that can be obtained by introducing $v$ in $S(W_{X'},X')$. We distinguish two cases. 

\smallskip\noindent\textbf{Case 1:} $N_X(v) = \emptyset$. For each active boundary $B_f$ of $S(W_{X'},X')$, we generate a new embedded graph by adding the closed walk $\langle v \rangle$ to $B_f$. 

\smallskip\noindent\textbf{Case 2:} $N_X(v) \neq \emptyset$. We look for a face $f^*$ of  $S(W_{X'},X')$ that contains all the vertices of $N_X(v)$ on its active boundary $B_f$ (which may consist of multiple closed walks). If such a face does not exist, we  discard $S(W_{X'},X')$.  Else, for each such face, we generate a set of entries $E_{f^*}$ as follows. Intuitively, we will insert $v$ inside $f^*$ and generate one entry of $E_{f^*}$ for each possible way in which $v$ can be connected to its neighbors. Namely, we can connect $v$ to its neighbors by means of different intersection vertices and by realizing different permutations of the edges around $v$ and around those neighbors that appear multiple times along some closed walk of $B_f$; refer to \cref{fig:sketched-2} for an illustration. Concerning the intersection vertices, we can use those that already belong to $B_f$ and are adjacent only to vertices in $N_X(v)$, as well as we can create new ones. We note that since $v$ has at most $\omega-1$ neighbors in $N_X(v)$, there are $\sum_{i=1}^{\omega-1} {\omega-1 \choose i} = 2^{\omega-1}$ possible combinations of intersection vertices (see also the proof of \cref{lem:number}). This is done avoiding inessential intersection vertices and twin-pairs. For each choice of intersection vertices, since the degree of a vertex is $O(\omega)$, there are $\omega^{O(\omega)}$ distinct rotation systems to consider. Additionally, if $B_f$ consists of multiple closed walks, we shall consider all possible permutations of the edges around $v$ that do not cause edge crossings (i.e., any edge permutation in which there are no four edges $e_1,e_2,e_3,e_4$ in this order around $v$, such that $e_1,e_3$ connect $v$ to the vertices of a closed walk $\pi$ and $e_2,e_4$ connect $v$ to the vertices of a closed walk $\pi'$ with $\pi \neq \pi'$), and we consider each of them independently as a new embedded graph. Based on the fixed intersection vertices and rotation system, if the insertion of $v$ does not split $f^*$ into multiple faces, we can suitably update~$B_f$, otherwise we can generate the new active boundaries that appear in place of $B_f$; see in particular \cref{fig:sketched-g}. Also, for each newly introduced edge $e$ in a closed walk, we set $c(e)=1$.

\begin{figure}[tb]
\centering
\begin{subfigure}{.24\textwidth}
\centering
\includegraphics[page=6,width=\textwidth]{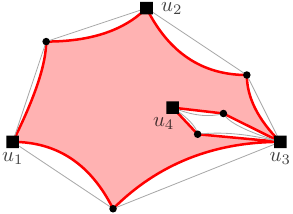}
\subcaption{$S(W_{X'},X')$\label{fig:sketched-c}}
\end{subfigure}\hfill
\begin{subfigure}{.24\textwidth}
\centering
\includegraphics[page=8,width=\textwidth]{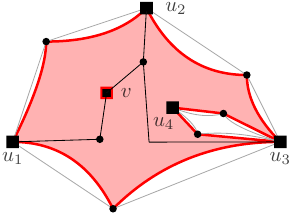}
\subcaption{\label{fig:sketched-e}}
\end{subfigure}\hfill
\begin{subfigure}{.24\textwidth}
\centering
\includegraphics[page=9,width=\textwidth]{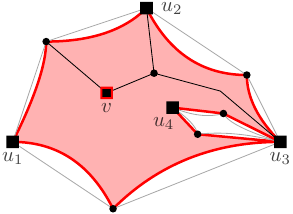}
\subcaption{\label{fig:sketched-f}}
\end{subfigure}\hfill
\begin{subfigure}{.24\textwidth}
\centering
\includegraphics[page=11,width=\textwidth]{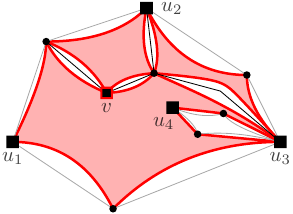}
\subcaption{$S(W_X,X)$\label{fig:sketched-g}}
\end{subfigure}
\caption{Illustration for the addition of vertex $v$. (a) Details of a face of $S(W_{X'},X')$ that contains all the neighbors of  $v$.  (b--c) Two distinct embedded graphs computed from $S(W_{X'},X')$ by introducing vertex~$v$ in different ways. (d) The \sketch $S(W_X,X)$ obtained by replacing the active boundary of the red face  with the new active boundaries corresponding to the three newly created active faces in (c).
\label{fig:sketched-2}}
\end{figure} 

\medskip\noindent\textbf{Merge operation.} 
Let $X$ be a join bag whose children $X_1$ and $X_2$ in $T$ have $k$-map (hole-free) \feasible records $R_{X_1}$ and $R_{X_2}$, respectively.
We generate $R_X$ from $R_{X_1}$ and $R_{X_2}$. Since $X$ is a join bag, $X$, $X_1$, and $X_2$ contain the same vertices, whereas $G_{X_1}$ and $G_{X_2}$ only share the vertices in $X$. Consider any pair of \sketches $S(W_{X_1},X)$ of $R_{X_1}$ and $S(W_{X_2},X)$ of $R_{X_2}$. Such \sketches share the same set of real vertices, whereas they may have different sets of intersection vertices and different combinatorial embeddings. At high-level, we aim at combining  $S(W_{X_1},X)$ and $S(W_{X_2},X)$ in all possible ways, provided that the original rotation and position systems of each \sketch are preserved and that we never insert a subgraph of one \sketch into a non-active face of the other. In practice, we apply the \emph{merge operation}, consisting of the next steps. 

\begin{enumerate}[(S.1)]
\item\label{s:1} We compute all possible unions of the two abstract graphs underlying the two \sketches. Namely, let $I_{X_1}$ and $I_{X_2}$ be the sets of intersection vertices of $S(W_{X_1},X)$ and $S(W_{X_2},X)$, respectively. We identify each pair of real vertices the two \sketches share, and we consider all possible abstract graphs whose set of intersection vertices $I_X$ is such that: (a) $I_X \subseteq I_{X_1} \cup I_{X_2}$; (b) for each intersection vertex of $I_{X_1}$ there is an intersection vertex in $I_X$ with the same set of neighbors, and the same holds for $I_{X_2}$.

\item\label{s:2} For each generated graph $S^*$, we compute all combinatorial embeddings, i.e., all possible rotation and position systems yielding a topological embedding on the sphere of $S^*$. If no such  combinatorial embeddings exist, we discard $S^*$, else we go to the next step.

\item\label{s:3} We generate all possible one-to-one mappings $\phi_1$ between intersection vertices of $S^*$ and of $S(W_{X_1},X)$, and all possible one-to-one mappings $\phi_2$ between intersection vertices of $S^*$ and of $S(W_{X_2},X)$. 

\item\label{s:4} We check, for each pair $\phi_1,\phi_2$, that the restriction of the resulting embedded graph on the real vertices, intersection vertices (up to the mapping defined by $\phi_1$ and $\phi_2$) and edges of each of the two \sketches preserves the corresponding rotation and position systems. If so, we go to the next step; otherwise, we discard the candidate solution.

\item\label{s:5} Since the previous step guaranteed that the active boundaries of each \sketch are preserved when looking at the corresponding restriction, we can verify that there is no subgraph of one \sketch inside a non-active face of the other.

\item\label{s:6} We suitably update the active boundaries of the resulting embedded graph and we add it to $R_X$. More precisely, the boundary of a face is active if it does not correspond to a non-active boundary in any of the two \sketches and it contains either exactly one anchor vertex or at least two anchor vertices.

\item\label{s:7} We remove inessential intersection vertices and iteratively one intersection vertex for each twin-pair, until there are no twin-pairs. 

\item\label{s:8} Once all pairs of \sketches have been processed, we remove possible~duplicates.
\end{enumerate}

\medskip

This concludes the description of the main algorithmic steps for proving \cref{th:main}. Next, we provide lemmas to establish the correctness and the time complexity of these steps.

\begin{lemma}\label{lem:forget}
Let $X$ be a forget bag whose child $X'$ in $T$ has a $k$-map (resp.\ hole-free) \feasible record $R_{X'}$. 
The algorithm either rejects the instance or computes a $k$-map (resp.\ hole-free) \feasible record $R_X$ of~$X$~in~$\omega^{O(\omega)}$~time.
\end{lemma}
\begin{proof}
Let $v$ be the vertex forgotten by $X$. We prove that the record $R_X$ generated by applying the deletion operation is \feasible, given that $R_{X'}$ is \feasible. 
In particular, since we removed possible duplicates, F\ref{f3} holds and it remains to argue about F\ref{f1} and F\ref{f2}. To this aim, since $X$ is a forget bag, note that $G_X = G_{X'}$. Hence any \compact witness $W_{X'}$ of $G_{X'}$ is also a \compact witness of~$G_X$. Moreover, since $R_{X'}$ is \feasible, it follows by F\ref{f1} that $R_{X'}$ contains a \sketch~$S(W_{X'},X')$ for every \compact witness~$W_{X'}$.  Now since $X' = X\cup\{v\}$, the \sketch  of $W_{X'}$ with respect to $X$, namely~$S(W_{X'},X)$, coincides with the one obtained by applying the deletion operation to~$S(W_{X'},X')$. Thus F\ref{f1} holds for~$X$. Similarly, since $R_{X'}$ is \feasible, it follows by F\ref{f2} that every entry of $R_{X'}$ is the \sketch $S(W_{X'},X')$  of a \compact witness $W_{X'}$ of $G_{X'}$.  Again since $X' = X\cup\{v\}$, the entry of $R_X$ obtained by applying the deletion operation to $S(W_{X'},X')$ corresponds to the \sketch $S(W_{X'},X)$. Thus F\ref{f2} holds for $X$ and consequently $R_X$ is \feasible, as claimed. Suppose now that $R_{X'}$ is $k$-map \feasible, i.e, $R^*_{X'} \neq \emptyset$. We show how to check whether a \sketch of $R_X$ belongs to $R^*_X$. Since the deletion operation does not modify the degree of any intersection vertex, the subrecord $R^*_X$  contains all \sketches of $R_X$ generated from \sketches in $R^*_{X'}$. Based on this observation, we can check whether $R^*_{X} = \emptyset$ or not. In the former case the algorithm rejects the instance, in the latter case $R_X$ is $k$-map \feasible. Suppose that $R_{X'}$ is hole-free \feasible, i.e., $R^\circ_{X'} \neq \emptyset$. Again the subrecord $R^\circ_X$  contains all \sketches of $R_X$ that have been generated from \sketches in $R^\circ_{X'}$ and that contain no active boundary whose edge counters sum up to 4. To decide whether an active boundary is complete, it suffices to check  whether the parent of $X$ is a forget bag such that the shortcuttings due to the removal of the forgotten vertex make that active boundary a self-loop. If any complete active boundary does not meet this condition, the corresponding \sketch does not~belong~$R^\circ_X$. As before if $R^\circ_{X} = \emptyset$ the algorithm rejects the instance, otherwise $R_X$ is hole-free \feasible.

By \cref{le:rx}, $R_{X'}$ contains $\omega^{O(\omega)}$ entries, each of size $O(\omega)$. Updating each of them takes $O(\omega)$ time. Also, $R_{X}$ contains at most as many entries as $R_{X'}$. It follows that removing duplicates can be naively done in $(\omega^{O(\omega)})^2 \in \omega^{O(\omega)}$ time. For the sake of efficiency, if we interpret each rotation and position system together as a number with $\tilde{O}(\omega^2)$ bits, then removing duplicates can be done in $\tilde{O}(\omega^2) \cdot \omega^{O(\omega)}\in \omega^{O(\omega)}$ time by using radix sort (we omit the details as the asymptotic running time would be the same). We have seen that condition F\ref{f4} is always verified. Checking condition F\ref{f5} requires scanning each active boundary in $R_X$ and decide whether it is complete or not, and if so to verify whether it will become a self-loop when visiting the parent of $X$. This can be done in $O(\omega)$ time for each of the $O(\omega)$ active boundaries of each of the $\omega^{O(\omega)}$ \sketches, and thus in $\omega^{O(\omega)}$ time overall. Thus~$R_X$ and its subrecords can be computed in $\omega^{O(\omega)}$ time, as desired.\end{proof}

\begin{lemma}\label{lem:introduce}
Let $X$ be an introduce bag whose child $X'$ in $T$ has a $k$-map (resp.\ hole-free) \feasible record $R_{X'}$. 
The algorithm either rejects the instance or computes a $k$-map (resp.\ hole-free) \feasible record $R_X$ of~$X$~in~$\omega^{O(\omega)}$~time.
\end{lemma}
\begin{proof}
Let $v$ be the vertex introduced by $X$. We prove that the record $R_X$ generated by applying the addition operation is \feasible, given that $R_{X'}$ is \feasible. 
Regarding F\ref{f1}, let $W_{X'}$ and $W_X$ be a witness of $G_{X'}$ and $G_X$, respectively, such that $W_X[G_{X'}]=W_{X'}$. Since F\ref{f1} holds for $R_{X'}$, we know that $S(W_{X'},X') \in R_{X'}$. Observe that the only difference between $W_X$ and $W_{X'}$ lies in the presence of vertex $v$ and of a (possibly empty) set $I_v$ of intersection vertices adjacent to $v$.

If $N_X(v) = \emptyset$, then $v$ forms a trivial closed walk that might be added in any face of $W_{X'}$ that either consists of exactly one anchor vertex or contains at least two anchor vertices (among possibly other non-anchor vertices). 
We recall that an active face satisfying the mentioned properties corresponds to an active boundary of the witness' \sketch. 
Also, adding the closed walk to a face that contains more than one vertex, but at most one anchor vertex, on its boundary would imply that the resulting witness cannot be augmented to a witness of $G$, since $G$ is biconnected. Since Case 1 places $v$ in all possible active boundaries of $S(W_{X'},X')$, we can conclude that $S(W_X,X)$ belongs to $R_X$. 

On the other hand, if $N_X(v) \neq \emptyset$,  then all $v$'s neighbors belong to a common boundary of some face $f$ of $W_{X'}$, as otherwise the rotation system of $W_X$ would not be compatible with a topological embedding (in particular, some edges would cross each other). Hence all $v$'s neighbors are part of the same active boundary $B_f$ of $S(W_{X'},X')$. Since Case 2 exhaustively considers all ways in which $v$ can be inserted into $B_f$,  avoiding inessential intersection vertices and twin-pairs (which cannot belong to $W_X$ since it is \compact), we can again conclude that $S(W_X,X)$ belongs to $R_X$.  Consequently F\ref{f1} holds for $R_X$. 

About F\ref{f2}, it suffices to prove that each entry generated by the addition operation is indeed a \sketch of some \compact witness of $G_X$ with respect to $X$. Since F\ref{f2} holds for $R_{X'}$, the addition operation starts from a \sketch $S(W_{X'},X')$ and it generates new entries in which there are neither inessential intersection vertices nor twin-pairs; therefore, such entries are indeed \sketches of \compact witnesses, as desired. 

Concerning F\ref{f3}, if $R_X$ contained two entries $r_1,r_2$ that are the same (up to a homeomorphism of the sphere), then $r_1$ and $r_2$ would have been originated by the same \sketch $r$ of $R_{X'}$, as otherwise either $r_1$ and $r_2$ would not be the same or F\ref{f3} would not hold for $R_{X'}$. On the other hand, since the addition operation inserts $v$ in different ways but without repetitions, it cannot generate two entries that are the same starting from a single entry of $R_{X'}$. Thus F\ref{f3} holds for $R_{X}$. 

If $R_{X'}$ is $k$-map \feasible, we know that $R^*_X$ contains those \sketches of $R^*_{X'}$ for which the addition operation did not introduce intersection vertices of degree larger than $k$. Based on this observation, we can check whether $R^*_{X} = \emptyset$ or not. In the former case the algorithm rejects the instance, in the latter case $R_X$ is $k$-map \feasible.  The case when $R_{X'}$ is hole-free \feasible can be proved analogously as in the proof of \cref{lem:forget}.

Finally, each single entry constructed by the addition operation can be computed in $O(\omega)$ time and  $R_X$ contains $\omega^{O(\omega)}$ entries by \cref{le:rx}. Also, condition F\ref{f4} can be easily verified in $O(\omega)$ time, for each of the $\omega^{O(\omega)}$ \sketches of $R_X$. Checking condition F\ref{f5} requires scanning each active boundary in $R_X$ and decide whether it is complete or not. This can be done in $O(\omega)$ time, for each of the $O(\omega)$ active boundaries of each of the $\omega^{O(\omega)}$ \sketches, and thus in $\omega^{O(\omega)}$ time overall. Thus $R_X$ and its subrecords can be computed in $\omega^{O(\omega)}$ time.
\end{proof}

The proof of the next lemma exploits the merge operation.

\begin{lemma}\label{lem:join}
Let $X$ be a join bag whose children $X_1$ and $X_2$ in $T$ both have $k$-map (resp.\ hole-free) \feasible records $R_{X_1}$ and $R_{X_2}$. The algorithm either rejects the instance or computes a $k$-map (resp.\ hole-free) \feasible record $R_X$ of $X$ in $\omega^{O(\omega)}$ time.
\end{lemma}
\begin{proof}
We prove that the record $R_X$ generated by applying the merge operation is \feasible, given that $R_{X_1}$ and $R_{X_2}$ are \feasible.
Consider any \compact witness $W_X$ of $G_X$ and its restrictions $W_X[G_{X_1}]$ and $W_X[G_{X_2}]$ to $G_{X_1}$ and $G_{X_2}$, respectively. By definition of restriction, there must exist a mapping of the intersection vertices of $W_X$ to the intersection vertices of $W_X[G_{X_1}]$ such that when looking at the restriction of $W_X$ to the real and intersection vertices of $W_X[G_{X_1}]$ (up to the above mentioned mapping), the rotation and position systems of $W_X[G_{X_1}]$ are preserved. The same property must hold for $W_X[G_{X_2}]$. These properties clearly carry over to the corresponding \sketches $S(W_X,X)$,  $S(W_X[G_{X_1}],X)$, and $S(W_X[G_{X_2}],X)$. Since $R_{X_1}$ and $R_{X_2}$ are \feasible, they contain $S(W_X[G_{X_1}],X)$ and $S(W_X[G_{X_2}],X)$, respectively. Hence, Steps S.\ref{s:1}--S.\ref{s:4} guarantee that the aforementioned mapping is considered and that all the above properties hold on the candidate solutions given by the combination of $S(W_X[G_{X_1}],X)$ and $S(W_X[G_{X_2}],X)$. Moreover, any subgraph of $W_X$ that belongs to $W_X[G_{X_1}]$ but not to $W_X[G_{X_2}]$, except for the shared vertices of~$X$, must lie in an active face of $W_X[G_{X_2}]$ (and vice-versa); if this is not the case, then $W_X$ would not be augmentable to a witness of $G$, since $G$ is biconnected. This property translates into verifying that any subgraph of $S(W_X[G_{X_1}],X)$ lies in an active face of $S(W_X[G_{X_2}],X)$ (and vice-versa). This is achieved in Step S.\ref{s:5}. Step  S.\ref{s:6} suitably updates the active boundaries so that a boundary is active only if it represents a face of $W_X$ that either consists of exactly one anchor vertex or contains at least two anchor vertices, as by definition of active boundary. Step  S.\ref{s:7} removes inessential intersection vertices and twin-pairs, which is a safe operation because $W_X$ is \compact. Therefore we can conclude that $S(W_X,X)$ belongs to $R_X$, and thus F\ref{f1} holds for $R_X$.   Concerning F\ref{f2}, any entry $S$ in $R_X$ generated by the merge operation, starting from entries $S(W_{X_1},X)\in R_{X_1}$ and $S(W_{X_2},X)\in R_{X_2}$, defines a way to combine the combinatorial embeddings of $S(W_{X_1},X)$ and $S(W_{X_2},X)$ at common real vertices and at possibly common (based on some mappings $\phi_1$ and $\phi_2$) intersection vertices. Such information can be used to combine in the same way the corresponding witnesses $W_{X_1}$ and $W_{X_2}$, which exist because F\ref{f2} holds for $R_{X_1}$ and $R_{X_2}$, respectively. On the other hand, such combination yields a \compact witness $W_X$ of $G_X$ with respect to $X$, whose \sketch is $S$, as desired. Thus F\ref{f2} holds for $R_X$. In Step S.\ref{s:8} we remove possible duplicates, hence F\ref{f3} holds by construction for $R_X$. Therefore $R_X$ is \feasible. Since the merge operation does not increase the degree of intersection vertices, and since  $R_{X_1}$ and $R_{X_2}$ are $k$-map \feasible, the subrecord $R^*_X$  contains all \sketches of $R_X$ generated from \sketches in $R^*_{X_1}$ and $R^*_{X_2}$. If  $R^*_X=\emptyset$, the algorithm rejects the instance, otherwise $R_X$ is $k$-map \feasible. If $R_{X_1}$ and $R_{X_2}$ are hole-free \feasible, $R^\circ_X$  contains all \sketches of $R_X$ that are generated from \sketches in $R^\circ_{X_1}$ and $R^\circ_{X_2}$ and whose complete active boundaries are such that the edge counters sum up to 4. If $R^\circ_X=\emptyset$, the algorithm rejects the instance, otherwise $R_X$ is hole-free \feasible.

Concerning the time complexity, we process each pair of \sketches, one in $R_{X_1}$ and one in $R_{X_2}$, and since both $R_{X_1}$ and $R_{X_2}$ are \feasible, we have $\omega^{O(\omega)}$ such pairs. Each of Steps S.\ref{s:1}, S.\ref{s:2}, and S.\ref{s:3} generates $\omega^{O(\omega)}$ new entries, and each entry is computed in $O(\omega)$ time. The remaining steps all run in $O(\omega)$ time for each processed entry.
Condition F\ref{f4} can be easily verified in $O(\omega)$ time, for each of the $\omega^{O(\omega)}$ \sketches of $R_X$. Furthermore, verifying condition F\ref{f5} requires scanning the active boundaries of each entry in $R_X$ and deciding whether it is complete or not. This can also be done in $O(\omega)$ time for each of the $O(\omega)$ active boundaries of each of the $\omega^{O(\omega)}$ \sketches, and thus in $\omega^{O(\omega)}$ time overall. Consequently, $R_X$ and its subrecords can be computed in $\omega^{O(\omega)}$ time.
\end{proof}

\noindent Lemmas~\ref{lem:forget}--\ref{lem:join} imply the next theorem, which summarizes the correctness of the approach.

\begin{theorem}\label{th:correct}
Let $G$ be a graph in input to the algorithm, along with a nice  tree-decomposition $(T,\cal X)$ of $G$ and an integer $k>0$. Graph $G$ is a $k$-map graph, respectively a hole-free $k$-map graph, if and only if the algorithm reaches the root $\rho$ of $T$ and the record $R_\rho$ is $k$-map~\feasible, respectively hole-free~\feasible. 
\end{theorem}

We are finally ready to prove \cref{th:main}. We recall that if $k \ge n-1$, recognizing  $n$-vertex (resp.\ hole-free) $k$-map graphs coincides with recognizing general $n$-vertex (resp.\ hole-free) map graphs.

\medskip\noindent\textbf{Proof of \cref{th:main}}. 
We first discuss the decision version of the problem for a fixed $k>0$. Namely, the algorithm described below is used in a binary search to find the optimal value of $k$. Recall that $t$ is the width of the tree decomposition (i.e., $\omega=t+1$). Note that, if $G$ is a positive instance, then $k$ varies in the range $[1,t+1]$, since the size of the largest clique of $G$ is at most $t+1$. Thus the algorithm is executed $O(\log t)$ times, which however does not affect the asymptotic running time.

If $G$ is not biconnected, by \cref{pr:connectivity}, it is not hole-free, and it is  $k$-map if and only if all its biconnected components are $k$-map. Hence we run our algorithm on each biconnected component independently. \cref{th:correct} implies the correctness of the algorithm (which assumes the input graph to be biconnected).

For the time complexity, suppose that $G$ has $h \ge 1$ biconnected components and let $n_i$ be the size of the $i$-th component $C_i$, for each $i \le h$. Decomposing $G$ into its biconnected components takes $O(n+m)$ time~\cite{DBLP:conf/focs/TarjanV84}, where $m$ is the number of edges of $G$ and, since $G$ has treewidth $t$, it holds $m \in O(n \cdot t^2)$. Given a tree-decomposition of $G$ with $O(n)$ nodes and width $t$, we can easily derive a tree-decomposition $(T_i, \mathcal{X}_i)$ for each $C_i$ in overall $O(n)$ time, such that each $T_i$ has $O(n_i)$ nodes and width at most $t$. Then we can apply the algorithm in~\cite{DBLP:journals/jal/BodlaenderK96} to obtain, in $O(n_i)$-time, a nice tree-decomposition of $C_i$ with $O(n_i)$ nodes without increasing the original width. Since each bag is processed in $t^{O(t)}$ time by Lemmas~\ref{lem:forget}--\ref{lem:join}, the algorithm runs in $t^{O(t)} \cdot n_i$ time for each $C_i$. Since $\sum_{i=1}^h n_i \in O(n)$, decomposing the graph and applying the algorithm to all its biconnected components takes $t^{O(t)} \cdot n$ time. 

To reconstruct a witness of a yes-instance, we store additional pointers for each record (a common practice in dynamic programming). Namely, for each \sketch $S$ of a record $R_X$ of a bag $X$, we store a pointer to the \sketch of the child bag $X'$ that generated $S$, if $X$ is an introduce or forget bag, and we store two pointers to the two \sketches of the children bags $X_1$ and $X_2$ that generated $S$, if $X$ is a join bag. With these pointers at hand, we can apply a top-down traversal of $T$, starting at any \sketch of the non-empty subrecord  of $\rho$, and reconstruct the corresponding witness $W$ by incrementally combining the retrieved \sketches, except at forget bags (the only points in which we lose information). Suppose first that $G$ is a $k$-map graph but not hole-free. If $G$ is not biconnected, a witness $W^*$ of $G$ is obtained by merging the witnesses of its biconnected components. Note that distinct witnesses corresponding to distinct biconnected components of $G$ can only share real vertices. Thus, each intersection vertex of $W^*$ has degree at most $k$ and $W^*$ is a certificate by \cref{pr:kmapdegree}. Suppose now that $G$ is a hole-free $k$-map graph. Then $G$ is biconnected and the resulting witness is a biconnected quadrangulation whose intersection vertices have degree at most $k$, a certificate by \cref{th:holefree}.
\qed

\section{Conclusions and Open Problems}\label{se:conclusions}

We have shown how to recognize (hole-free) $k$-map graphs  in linear time for inputs having bounded treewidth. The general problem of recognizing map graphs efficiently remains a major algorithmic challenge. 
To restrict the complexity of the input,  further parameters of interest might be the cluster vertex deletion number~\cite{DBLP:journals/mst/HuffnerKMN10} and the clique-width~\cite{DBLP:journals/jcss/CourcelleER93} of the input graph, as well as the treewidth of the putative witness~\cite{DBLP:journals/disopt/MnichRS18}. 

Another interesting line of research would be  generalizing our framework to recognize \emph{$(g,k)$-map graphs}, i.e., those graphs that admit a $k$-map on a surface of genus $g$ (see, e.g., \cite{DBLP:journals/siamdm/DujmovicEW17}). 

We finally recall that the complexity of recognizing (hole-free) $k$-map graphs is open for any fixed $k \ge 5$. A natural step in this direction is hence studying the complexity of recognizing (hole-free) $5$-map graphs. 

\bibliographystyle{abbrv}
\bibliography{paper}

\include{appendix-arxiv}

\end{document}

%% file: appendix-arxiv.tex
\section*{Appendix}

\section*{Monadic second-order logic formulation}\label{apx:courcelle}

We prove that the problem of recognizing map graphs can be expressed by using MSO$_2$ logic, which implies the existence of a  fixed-parameter tractable algorithm for parameterized by treewidth.

\begin{theorem}\label{thm:mso}
Given an $n$-vertex graph $G$ of treewidth $t$, there is an algorithm that decides whether $G$ is a map graph in time $f(t) \cdot O(n)$, for some computable function $f$.
\end{theorem}
\begin{proof}
Let $V$ and $E$ be the vertex and edge set of $G$, respectively. We construct a  graph $G^*=(V \cup C,E \cup D)$ by augmenting $G$. For every subset $S$ of $V$ such that $S$ forms a clique in $G$: (i) We add a vertex $v_S \in C$ to $G^*$, and (ii) We add an edge $(v_S,u) \in D$ for each $u \in S$. Since $G$ has $n$ vertices and treewidth $t$, it admits a tree-decomposition $T$ with $O(n)$ bags, such that each bag contains at most $t$ vertices. Also, for any clique of $G$ there is a bag that contains all its vertices. Altogether, it follows that $G$ contains $O(2^t\,n)$ cliques and hence $G^*$ has $O(2^t\,n)$ vertices. Moreover, the treewidth of $G^*$ is at most $t+1$. Namely, we can obtain a valid tree-decomposition $T^*$ of $G^*$ from $T$ as follows. For a vertex $v_S \in C$, let $S$ be the corresponding clique in $V$ and let $\nu$ be a bag of $T$ that contains all the vertices of $S$. For each such a vertex $v_S$, we add a new leaf bag $\nu^*$ in $T^*$, connected only to $\nu$ and containing $v_S$ and all the vertices in $S$. It is immediate to verify that $T^*$ is a tree-decomposition of $G^*$ of width at most $t+1$. 

By construction, $G$ is a map graph if and only if there exists a subset $B$ of $D$ such that: (i) The graph $G_B$ formed by the edges of $B$ is planar (note that it is bipartite by construction), and (ii) For every edge $(u,v)$ of $E$, there is a path between $u$ and $v$ in $G^*$ composed of two edges of $B$. Indeed, if $B$ exists, then $G_B$ is a witness of $G$. Both conditions (planarity and the existence of a length-2 path) can be expressed in MSO$_2$ logic\footnote{B. Courcelle and J. Engelfriet. Graph Structure and Monadic Second-Order Logic - A Language-Theoretic Approach, volume 138 of Encyclopedia of mathematics and its applications. Cambridge University Press, 2012.}.  
Consequently, the statement follows by Courcelle's theorem~\cite{DBLP:journals/iandc/Courcelle90}. 
\end{proof}
 
 We remark that the proof of \Cref{thm:mso} can be easily modified to find the minimum $k$ such that $G$ is a $k$-map graph. Namely, for the decision version of the problem, it suffices to add a vertex $v_S \in C$ to $G^*$ only if the clique $S$ has size at most $k$. 
 However, it is less obvious how to adjust the proof in order to  test whether $G$ is also hole-free, in particular, how to additionally ensure that  $G_B$ has a planar embedding in which all faces have length at most six~\cite{DBLP:journals/jacm/ChenGP02}.

%% file: main-arxiv.bbl
\begin{thebibliography}{10}

\bibitem{DBLP:conf/swat/AngeliniBLGMT22}
P.~Angelini, M.~A. Bekos, G.~{Da Lozzo}, M.~Gronemann, F.~Montecchiani, and
  A.~Tappini.
\newblock Recognizing map graphs of bounded treewidth.
\newblock In A.~Czumaj and Q.~Xin, editors, {\em {SWAT} 2022}, volume 227 of
  {\em LIPIcs}, pages 8:1--8:18. Schloss Dagstuhl - Leibniz-Zentrum f{\"{u}}r
  Informatik, 2022.

\bibitem{DBLP:journals/jgaa/AngeliniLBFPR17}
P.~Angelini, G.~{Da Lozzo}, G.~{Di Battista}, F.~Frati, M.~Patrignani, and
  I.~Rutter.
\newblock Intersection-link representations of graphs.
\newblock {\em J. Graph Algorithms Appl.}, 21(4):731--755, 2017.

\bibitem{DBLP:journals/corr/abs-2204-11495}
M.~A. Bekos, G.~{Da Lozzo}, P.~Hlinen{\'{y}}, and M.~Kaufmann.
\newblock Graph product structure for h-framed graphs.
\newblock {\em CoRR}, abs/2204.11495, 2022.

\bibitem{DBLP:journals/siamcomp/Bodlaender96}
H.~L. Bodlaender.
\newblock A linear-time algorithm for finding tree-decompositions of small
  treewidth.
\newblock {\em {SIAM} J. Comput.}, 25(6):1305--1317, 1996.

\bibitem{DBLP:journals/jal/BodlaenderK96}
H.~L. Bodlaender and T.~Kloks.
\newblock Efficient and constructive algorithms for the pathwidth and treewidth
  of graphs.
\newblock {\em J. Algorithms}, 21(2):358--402, 1996.

\bibitem{DBLP:journals/dam/Brandenburg19}
F.~J. Brandenburg.
\newblock Characterizing 5-map graphs by 2-fan-crossing graphs.
\newblock {\em Discret. Appl. Math.}, 268:10--20, 2019.

\bibitem{DBLP:journals/algorithmica/Brandenburg19}
F.~J. Brandenburg.
\newblock Characterizing and recognizing 4-map graphs.
\newblock {\em Algorithmica}, 81(5):1818--1843, 2019.

\bibitem{DBLP:journals/jal/Chen01a}
Z.~Chen.
\newblock Approximation algorithms for independent sets in map graphs.
\newblock {\em J. Algorithms}, 41(1):20--40, 2001.

\bibitem{DBLP:journals/jgt/Chen07}
Z.~Chen.
\newblock New bounds on the edge number of a $k$-map graph.
\newblock {\em J. Graph Theory}, 55(4):267--290, 2007.

\bibitem{DBLP:conf/stoc/ChenGP98}
Z.~Chen, M.~Grigni, and C.~H. Papadimitriou.
\newblock Planar map graphs.
\newblock In {\em {STOC}}, pages 514--523. {ACM}, 1998.

\bibitem{DBLP:journals/jacm/ChenGP02}
Z.~Chen, M.~Grigni, and C.~H. Papadimitriou.
\newblock Map graphs.
\newblock {\em J. {ACM}}, 49(2):127--138, 2002.

\bibitem{DBLP:journals/algorithmica/ChenGP06}
Z.~Chen, M.~Grigni, and C.~H. Papadimitriou.
\newblock Recognizing hole-free 4-map graphs in cubic time.
\newblock {\em Algorithmica}, 45(2):227--262, 2006.

\bibitem{DBLP:conf/soda/ChenHK99}
Z.~Chen, X.~He, and M.~Kao.
\newblock Nonplanar topological inference and political-map graphs.
\newblock In {\em {SODA}}, pages 195--204. {ACM/SIAM}, 1999.

\bibitem{DBLP:journals/queue/Cormode17}
G.~Cormode.
\newblock Data sketching.
\newblock {\em {ACM} Queue}, 15(2):60, 2017.

\bibitem{DBLP:journals/iandc/Courcelle90}
B.~Courcelle.
\newblock The monadic second-order logic of graphs. {I}. {R}ecognizable sets of
  finite graphs.
\newblock {\em Inf. Comput.}, 85(1):12--75, 1990.

\bibitem{DBLP:journals/jcss/CourcelleER93}
B.~Courcelle, J.~Engelfriet, and G.~Rozenberg.
\newblock Handle-rewriting hypergraph grammars.
\newblock {\em J. Comput. Syst. Sci.}, 46(2):218--270, 1993.

\bibitem{DBLP:books/sp/CyganFKLMPPS15}
M.~Cygan, F.~V. Fomin, L.~Kowalik, D.~Lokshtanov, D.~Marx, M.~Pilipczuk,
  M.~Pilipczuk, and S.~Saurabh.
\newblock {\em Parameterized Algorithms}.
\newblock Springer, 2015.

\bibitem{DBLP:journals/talg/DemaineFHT05}
E.~D. Demaine, F.~V. Fomin, M.~T. Hajiaghayi, and D.~M. Thilikos.
\newblock Fixed-parameter algorithms for ($k$, $r$)-center in planar graphs and
  map graphs.
\newblock {\em {ACM} Trans. Algorithms}, 1(1):33--47, 2005.

\bibitem{DBLP:journals/jcss/GiacomoLM22}
E.~{Di Giacomo}, G.~Liotta, and F.~Montecchiani.
\newblock Orthogonal planarity testing of bounded treewidth graphs.
\newblock {\em J. Comput. Syst. Sci.}, 125:129--148, 2022.

\bibitem{DBLP:series/mcs/DowneyF99}
R.~G. Downey and M.~R. Fellows.
\newblock {\em Parameterized Complexity}.
\newblock Monographs in Computer Science. Springer, 1999.

\bibitem{DBLP:journals/siamdm/DujmovicEW17}
V.~Dujmovic, D.~Eppstein, and D.~R. Wood.
\newblock Structure of graphs with locally restricted crossings.
\newblock {\em {SIAM} J. Discret. Math.}, 31(2):805--824, 2017.

\bibitem{DBLP:journals/jacm/DujmovicJMMUW20}
V.~Dujmovi\'c, G.~Joret, P.~Micek, P.~Morin, T.~Ueckerdt, and D.~R. Wood.
\newblock Planar graphs have bounded queue-number.
\newblock {\em J. {ACM}}, 67(4):22:1--22:38, 2020.

\bibitem{DBLP:conf/focs/FominLMS12}
F.~V. Fomin, D.~Lokshtanov, N.~Misra, and S.~Saurabh.
\newblock Planar f-deletion: Approximation, kernelization and optimal {FPT}
  algorithms.
\newblock In {\em {FOCS}}, pages 470--479. {IEEE}, 2012.

\bibitem{DBLP:conf/icalp/FominLP0Z19}
F.~V. Fomin, D.~Lokshtanov, F.~Panolan, S.~Saurabh, and M.~Zehavi.
\newblock Decomposition of map graphs with applications.
\newblock In {\em {ICALP}}, volume 132 of {\em LIPIcs}, pages 60:1--60:15.
  Schloss Dagstuhl - Leibniz-Zentrum f{\"{u}}r Informatik, 2019.

\bibitem{DBLP:conf/soda/FominLS12}
F.~V. Fomin, D.~Lokshtanov, and S.~Saurabh.
\newblock Bidimensionality and geometric graphs.
\newblock In {\em {SODA}}, pages 1563--1575. {SIAM}, 2012.

\bibitem{DBLP:journals/mst/HuffnerKMN10}
F.~H{\"{u}}ffner, C.~Komusiewicz, H.~Moser, and R.~Niedermeier.
\newblock Fixed-parameter algorithms for cluster vertex deletion.
\newblock {\em Theory Comput. Syst.}, 47(1):196--217, 2010.

\bibitem{DBLP:conf/soda/JansenLS14}
B.~M.~P. Jansen, D.~Lokshtanov, and S.~Saurabh.
\newblock A near-optimal planarization algorithm.
\newblock In {\em {SODA}}, pages 1802--1811. {SIAM}, 2014.

\bibitem{DBLP:books/sp/Kloks94}
T.~Kloks.
\newblock {\em Treewidth, Computations and Approximations}, volume 842 of {\em
  LNCS}.
\newblock Springer, 1994.

\bibitem{DBLP:journals/algorithmica/KociumakaP19}
T.~Kociumaka and M.~Pilipczuk.
\newblock Deleting vertices to graphs of bounded genus.
\newblock {\em Algorithmica}, 81(9):3655--3691, 2019.

\bibitem{DBLP:journals/disopt/MnichRS18}
M.~Mnich, I.~Rutter, and J.~M. Schmidt.
\newblock Linear-time recognition of map graphs with outerplanar witness.
\newblock {\em Discret. Optim.}, 28:63--77, 2018.

\bibitem{DBLP:journals/jal/RobertsonS86}
N.~Robertson and P.~D. Seymour.
\newblock Graph minors. {II.} algorithmic aspects of tree-width.
\newblock {\em J. Algorithms}, 7(3):309--322, 1986.

\bibitem{DBLP:conf/focs/TarjanV84}
R.~E. Tarjan and U.~Vishkin.
\newblock Finding biconnected components and computing tree functions in
  logarithmic parallel time (extended summary).
\newblock In {\em {FOCS}}, pages 12--20. {IEEE}, 1984.

\bibitem{DBLP:conf/focs/Thorup98}
M.~Thorup.
\newblock Map graphs in polynomial time.
\newblock In {\em {FOCS}}, pages 396--405. {IEEE}, 1998.

\end{thebibliography}
